\documentclass[final,12pt]{article}
\usepackage[colorlinks=true,linkcolor=myblue,citecolor=mymagenta]{hyperref}
\usepackage{amssymb,amsmath,amsthm}
\usepackage{cite}
\usepackage{enumerate}
\usepackage{physics}
\usepackage[conditional,light,first,bottomafter]{draftcopy}
\draftcopyName{DRAFT\space\today}{130}
\draftcopySetScale{65}
\usepackage{lmodern}
\usepackage{microtype}
\usepackage{graphicx}
\usepackage{xcolor}
\usepackage{calc}
\newlength{\depthofsumsign}
\newlength{\depthofprodsign}
\definecolor{mymagenta}{RGB}{188,21,108}
\definecolor{myblue}{RGB}{38,52,150}

\newcommand{\nsum}{
    \mathop{%
        \raisebox
            {-1.4\depthofsumsign+1.4\depthofsumsign}
            {\scalebox
                {1.4}
                {$\displaystyle\sum$}%
            }
    }
}
\newcommand{\msum}{
    \mathop{%
        \raisebox
            {-1.1\depthofsumsign+1.2\depthofsumsign}
            {\scalebox
                {1.2}
                {$\displaystyle\sum$}%
            }
    }
}
\newcommand{\mprod}{
    \mathop{%
        \raisebox
            {-1.1\depthofprodsign+1.2\depthofprodsign}
            {\scalebox
                {1.2}
                {$\displaystyle\prod$}%
            }
    }
}

\usepackage[letterpaper,hmargin=3cm,vmargin=2.8cm]{geometry}
\usepackage{setspace}
\makeatletter
\renewcommand{\section}{\@startsection%
{section}%
{1}%
{0em}%
{1.7em}%
{1.2em}%
{\normalfont\large\centering\bfseries}}
\renewcommand{\@seccntformat}[1]%
{\csname the#1\endcsname.\hspace{0.5em}}
\renewcommand{\subsection}{\@startsection%
{subsection}%
{1}%
{0em}%
{1.7em}%
{1.2em}%
{\normalfont\bfseries}}
\makeatother

\numberwithin{equation}{section}
\newtheorem{theorem}{Theorem}[section]
\newtheorem{proposition}{Proposition}[section]
\newtheorem{lemma}{Lemma}[section]
\newtheorem{corollary}{Corollary}[section]
\theoremstyle{definition}
\newtheorem{definition}{Definition}
\newtheorem{hypothesis}{Hypothesis}
\newtheorem{remark}{Remark}

\newcommand{\inner}[2]{\left\langle#1,#2\right\rangle}
\newcommand{\cc}[1]{\overline{#1}}
\newcommand{\reals}{\mathbb{R}}
\newcommand{\nats}{\mathbb{N}}
\newcommand{\complex}{\mathbb{C}}
\newcommand{\ints}{\mathbb{Z}}
\newcommand{\convergesto}[1]{\xrightarrow[#1\to\infty]{}}
\newcommand{\wconvergesto}[1]{\xrightarrow[#1\to\infty]{w}}
\newcommand{\I}{{\rm i}}
\newcommand{\ie}{\emph{i.e.\kern.2em}}
\newcommand{\cf}{\emph{cf.\kern.2em}}
\newcommand{\viz}{\emph{viz.\kern.2em}}
\newcommand{\eg}{\emph{e.g.\kern.2em}}

\DeclareMathOperator{\dom}{dom}

\DeclareMathOperator{\dist}{dist}

\DeclareMathOperator{\diag}{diag}
\begin{document}
\begin{titlepage}
\title{Green matrix estimates of block Jacobi matrices II:
  Bounded gap in the essential spectrum
\footnotetext{%
Mathematics Subject Classification(2010):
41A10  
47B36, 
33E30  
}
\footnotetext{%
Keywords:
Block Jacobi operators;
Generalized eigenvectors;
Decay bounds.
}
}
\author{
\textbf{Jan Janas}
\\
\small Institute of Mathematics\\[-1.7mm]
\small Polish Academy of Sciences (PAN)\\[-1.7mm]
\small Ul. Sw. Tomasza 30, 31-027\\[-1.7mm]
\small Krakow, Poland\\[-1.7mm]
\small \texttt{najanas@cyf-kr.edu.pl}
\\[2mm]
\textbf{Sergey Naboko}
\\
\small Department of Mathematical Physics\\[-1.7mm]
\small Institute of Physics\\[-1.7mm]
\small St. Petersburg State University\\[-1.7mm]
\small Ulyanovskaya 1, St. Petersburg 198904, Russia\\[-1.7mm]
\small \texttt{naboko@snoopy.phys.spbu.ru}
\\[2mm]
\textbf{Luis O. Silva}\thanks{Supported by UNAM-DGAPA-PAPIIT IN110818
  and SEP-CONACYT CB-2015 254062. On sabbatical leave from UNAM with
  the support of PASPA-DGAPA-UNAM}
\\
\small Department of Mathematical Sciences\\[-1.7mm]
\small University of Bath\\[-1.7mm]
\small Claverton Down, Bath BA2 7AY, U.K.\\[-1.7mm]
\small \&\\[-1.7mm]
\small Departamento de F\'{i}sica Matem\'{a}tica\\[-1.7mm]
\small Instituto de Investigaciones en Matem\'aticas Aplicadas y en Sistemas\\[-1.7mm]
\small Universidad Nacional Aut\'onoma de M\'exico\\[-1.7mm]
\small C.P. 04510, M\'exico D.F.\\[-1.7mm]
\small \texttt{silva@iimas.unam.mx}
}
\date{}
\maketitle
\vspace{-6mm}
\begin{center}
\begin{minipage}{5in}
  \centerline{{\bf Abstract}} \smallskip This paper provides decay bounds
  for Green matrices and generalized eigenvectors of block Jacobi
  operators when the real part of the spectral parameter lies in a
  bounded gap of the operator's essential spectrum. The case of the
  spectral parameter being an eigenvalue is also considered. It is
  also shown that if the matrix entries commute, then the estimates
  can be refined. Finally, various examples of block Jacobi operators
  are given to illustrate the results.
\end{minipage}
\end{center}
\thispagestyle{empty}
\end{titlepage}
\section{Introduction}
\label{sec:intro}

This work gives decay bounds for the entries of the Green matrix
corresponding to a self-adjoint block Jacobi operator when the real
part of the spectral parameter lies in a bounded gap of the operator's
essential spectrum. These estimates also show how fast the sequence of
generalized eigenvectors decays. In a previous paper \cite{MR3820400},
decay bounds for Green matrix entries were established for the case of
semibounded block Jacobi operators. It turns out that, on the one
hand, the technique used in \cite{MR3820400} for unbounded gaps in the
essential spectrum cannot be applied to the case of bounded gaps
without substantial modifications. On the other hand, the estimates in
the case of bounded gaps are completely different from the ones of
unbounded gaps.

The results of this work provide a refinement of the Combes-Thomas
method used for obtaining estimates of the Green function's decay. A
crucial fact of this refinement is the use of
Lemma~\ref{lem:janas-naboko-stolz} which was already implemented in
\cite{MR2480099,MR3028179,MR1935594} and is in a certain sense behind
the improvements of the original method (see \cite{MR1441595}).
Combes-Thomas type estimates (\cf \cite{MR1441595,MR0391792}) are used
in the analysis of some discrete random models (see
\cite{MR2509110}). We remark that block type random operators have
been studied earlier (\eg \cite{MR3210957,MR3077279,MR969209}) and
a Combes-Thomas type estimate for random block operators is given in
\cite[Lem.\,5.7]{MR3077279}.

The decay bounds of the generalized eigenvectors established in this
work have interesting applications to the study of spectral phase
transition phenomena (see for instance \cite{MR2480099}). Notably,
block Jacobi matrices permit more freedom in the construction of
models which exhibit multi-threshold spectral phase transitions.

In the Hilbert space $\mathcal{H}$ of square-summable sequences whose
elements belong in turn to a Hilbert space $\mathcal{K}$, we consider
the operator $J$ associated with a second order difference equation
with operator coefficients acting in $\mathcal{K}$ (see
Definition~\ref{def:j-nought}). The operator coefficients, $A_n$ and
$B_n=B_n^*$ ($n\in\nats$), are the entries of a block Jacobi matrix
and, consequently, $J$ is called a block Jacobi operator. It is
assumed that, for any $n\in\nats$, $A_n,B_n$ are bounded and defined
on the whole space $\mathcal{K}$, \ie they belong to
$B(\mathcal{K})$. Additionally, the operators $A_n,B_n$ ($n\in\nats$)
are required to satisfy conditions guaranteeing that $J=J^*$ on a
proper domain and the existence of a gap in the essential spectrum of
$J$.  Under these assumptions, let us paraphrase the statements of
Theorems~\ref{thm:finite-interval} and
\ref{thm:finite-interval-in-spectrum}. Denote by $G_{jk}(\zeta)$ the
($j,k$)-block entry of the resolvent's matrix at the point
$\zeta\in\complex$ (Definition~\ref{def:green-function}). If the real
part of $\zeta$ is in the gap of the essential spectrum of $J$, then
there are positive constants $C_1$ and $C_2$ such that
    \begin{equation*}
      \norm{G_{mj}(\zeta)}_{B(\mathcal{K})}\le
    C_1
   \exp(-\gamma(\zeta)\mkern-12mu\nsum\limits_{k=\min(m,j)}^{\max(m,j)-1}
\mkern-12mu\phi_\delta\left(\norm{A_k}_{B(\mathcal{K})}\right))
\end{equation*}
when $\zeta$ is not an eigenvalue and
\begin{equation*}
    \norm{u_m}_{\mathcal{K}}\le
C_2\exp(-\gamma(\zeta)\nsum_{k=1}^{m-1}
\phi_\delta\left(\norm{A_k}_{B(\mathcal{K})}\right))
\end{equation*}
when $\zeta$ is an eigenvalue. The function $\phi_{\delta}$ is given
by \eqref{eq:phi-delta} and corresponds to a regularization of the
reciprocal function. The function $\gamma$ is given by
\eqref{eq:gamma-noncommut-small-imaginary} and
\eqref{eq:gamma-big-imaginary}, and the constants $C_1$ and $C_2$ do
not depend on $m$ or $j$ (although they may depend on the spectral
parameter $\zeta$). Apart from depending on $\zeta$ and the size of
the gap, $\gamma$ depends on a parameter which permits certain
optimization (Remark~\ref{rem:delta-parameter}).

The decay bounds given above cover a wide range of block Jacobi
operators. In these estimates, we provide an explicit expression for
the decay coefficient, and establish how they depend on the
off-diagonal entries of the block Jacobi operator. The fact that
the bounds depend inversely on the growth of the off-diagonal entries of a
Jacobi matrix was already shown by Combes-Thomas type estimates.  The
results of this work, on the one hand, generalize to the case of block
matrix operators the corresponding decay bounds obtained in
\cite{MR2480099,MR3028179} for the scalar case. On the other hand,
our results extend the region of the spectral parameter so that it is a
complex non necessarily real number.

When the entries of the block Jacobi operator commute, a subsequent
refinement of the results given above can be obtained (see
Section~\ref{sec:green-funct-estim}). Indeed, under the assumptions on
$J$ considered above and the additional requirement that the entries
commute pairwise (see in Theorem~\ref{thm:finite-interval-commutation}
the precise statement), one has
    \begin{equation*}
 \norm{\exp(\gamma(\zeta)\mkern-12mu\nsum\limits_{k=\min(m,j)}^{\max(m,j)-1}
\mkern-12mu\phi_\delta(\abs{A_k}))G_{mj}(\zeta)}_{B(\mathcal{K})}
\le C
\end{equation*}
for $\zeta$ being such that its real part is in the gap and it is not
an eigenvalue. Here again $\gamma$ is given by
\eqref{eq:gamma-noncommut-small-imaginary} and
\eqref{eq:gamma-big-imaginary}, and the constant $C$ does not depend
on $m$ or $j$. In this case of commuting entries, a result
corresponding to the case when $\zeta$ is an eigenvalue is also
established
(Theorem~\ref{thm:finite-interval-commutation-in-spectrum}). We remark
that this \emph{operator} way of estimating the decay of generalized
eigenvectors make it possible to optimize the estimates of these
vectors along any spatial direction in $\mathcal{K}$ by taking into
account the matrix structure of the blocks.

Section~\ref{sec:examples} provides various examples which illustrate
the results of this work. The asymptotic of generalized eigenvectors
are obtained heuristically on the basis of methods which have been
used to find the asymptotics of solutions to difference equations. We
consider the examples of this section to be of independent interest.

Finally, in Section~\ref{sec:discr-vers-theor}, we briefly describe an
alternative approach to establishing decay bounds for block Jacobi
operators. This approach makes use of the discrete nature of the
problem and allows us to obtain more precise results in certain cases.

\section{Preliminaries}
\label{sec:preliminaries}

\subsection{Notation}
\label{sec:notation}

By $\mathcal{H}$ we denote a separable infinite dimensional Hilbert
space which is always decomposed as an infinite orthogonal sum:
  \begin{equation*}
    \mathcal{H}=\bigoplus_{m=1}^\infty\mathcal{K}_m\,,
  \end{equation*}
  where $\mathcal{K}_m=\mathcal{K}$ for all $m\in\nats$ and
  $\mathcal{K}$ is either an infinite or finite dimensional subspace
  of $\mathcal{H}$. A Hilbert space so decomposed is usually denoted by
$l_2(\nats,\mathcal{K})$.

The symbol $\norm{\cdot}$ is used to denote the norm in $\mathcal{H}$,
while the norm in $\mathcal{K}$ is denoted by
$\norm{\cdot}_{\mathcal{K}}$. $B(\mathcal{H}$) and $B(\mathcal{K})$
denote the spaces of bounded linear operators defined on the whole
space $\mathcal{H}$ and $\mathcal{K}$, respectively. The norms in
$B(\mathcal{H})$ and $B(\mathcal{K})$ are denoted by
$\norm{\cdot}_{B(\mathcal{H})}$ and $\norm{\cdot}_{B(\mathcal{K})}$,
respectively.

A vector $u$ in
$\mathcal{H}$ can be written as a sequence
\begin{equation}
  \label{eq:sequence}
  u=\{u_m\}_{m=1}^\infty\,,\qquad
u_m\in\mathcal{K}_m
\end{equation}
or as an infinite column vector, \ie
$u=(u_1,u_2,u_3,\dots)^{\intercal}$. Note that
\begin{equation*}
  \norm{u}_{\mathcal{H}}^2=\sum_{m=1}^\infty\norm{u_m}_{\mathcal{K}}^2\,.
\end{equation*}

Throughout this work, we use $I$ to denote the identity operator in
the spaces $\mathcal{H}$ and $\mathcal{K}$ since it will cause no
confusion to use the same letter for these operators.  The orthogonal
projector in $\mathcal{H}$ onto the subspace $\mathcal{K}_m$ is
denoted by $P_m$ while the symbol $\widetilde{P}_M$ stands for the
orthogonal projector onto
\begin{equation*}
  \bigoplus_{m=1}^M\mathcal{K}_m\,.
\end{equation*}

Given a closed, densely defined operator $T$ in a Hilbert space, we
denote by $\abs{T}$ the operator $(T^*T)^{1/2}$ given by applying the
functional calculus to the self-adjoint operator $T^*T$. Finally, for
any self-adjoint operator $A$, we denote its essential spectrum, \ie
the union of the continuous spectrum and the eigenvalues of infinite
multiplicity, by $\sigma_{ess}(A)$.

\subsection{Block Jacobi matrices}
\label{sec:block-jacobi-matr}

Let us turn to the definition of block Jacobi operators. For any
sequence \eqref{eq:sequence}, consider the second order difference
expressions
\begin{subequations}
  \label{eq:difference-expr}
\begin{align}
   \label{eq:difference-recurrence}
  (\Upsilon u)_k&:= A_{k-1}^* u_{k-1} + B_k u_k + A_ku_{k+1}
  \quad k \in \mathbb{N} \setminus \{1\},\\
   \label{eq:difference-initial}
   (\Upsilon u)_1&:=B_1 u_1 + A_1 u_2\,,
\end{align}
\end{subequations}
where $B_k=B_k^*$ and $A_k$ are in $B(\mathcal{K})$ for any $k\in\mathbb{N}$.
\begin{definition}
 \label{def:j-nought}
  In $\mathcal{H}$, define the operator $J$ such that
  \begin{equation*}
    \dom(J):=\{u\in\mathcal{H}:\Upsilon
  u\in\mathcal{H}\}
  \end{equation*}
 and $Ju=\Upsilon u$ for any $u\in\dom(J)$. Since $\dom(J)$ is dense
 in $\mathcal{H}$, the adjoint of $J$, $J^*$, is an operator. One verifies
 from the difference expression that $J^*\subset J$. The operator $J$
 is said to be defined in the maximal domain given by the difference
 expression $\Upsilon$, whereas $J^*$ is the minimal closed operator
 associated with $\Upsilon$ (see \cite[Chap.\,7 Sec.\,2.5]{MR0222718}).
\end{definition}
The block tridiagonal matrix
\begin{equation}
  \label{eq:block-jm}
\begin{pmatrix}
B_1&A_1&0&0&\cdots\\
A_1^*&B_2&A_2&0&\\
0&A_2^*&B_3&A_3&\ddots\\
0&0&A_3^*&B_4&\ddots\\
\vdots&&\ddots&\ddots&\ddots
\end{pmatrix}
\end{equation}
can be regarded as the matrix representation of the operator $J^*$
when every element $u$ in $\mathcal{H}$ is written as in \eqref{eq:sequence}.
(\cf \cite[Sec. 47]{MR1255973} where the matrix
representation of an unbounded symmetric operator is given).
\begin{remark}
\label{rem:selfadjointness-carleman}
In this work, we impose conditions on the sequences
$\{A_m\}_{m=1}^\infty$ and $\{B_m\}_{m=1}^\infty$ so that $J$ is
self-adjoint. A sufficient condition for this to happen is the
generalized Carleman criterion \cite[Chap.\,7 Thm. 2.9]{MR0222718},
viz., if $\sum_{m=1}^\infty 1/\norm{A_m}_{B(\mathcal{K})}=+\infty$,
then $J$ is self-adjoint. A particular case of an operator satisfying
the Carleman criterion is a bounded block Jacobi operator for which
the sequences $\{\norm{A_m}_{B(\mathcal{K})}\}_{m=1}^\infty$ and
$\{\norm{B_m}_{B(\mathcal{K})}\}_{m=1}^\infty$ are bounded \cite[Chap.\,7
Sec.\,2.11]{MR0222718}.
\end{remark}
Sometimes, the same notation is used for the matrix and the operator. This
cannot lead to confusion when $J=J^*$ (in this case the maximal and
minimal domains coincide) and in the case of diagonal matrices.
Thus,
\begin{equation*}
  \diag\{C_m\}_{m=1}^\infty:=
\begin{pmatrix}
C_1&0&0&0&\cdots\\
0&C_2&0&0&\\
0&0&C_3&0&\ddots\\
0&0&0&C_4&\ddots\\
\vdots&&\ddots&\ddots&\ddots
\end{pmatrix}\,,
\end{equation*}
where $C_m\in B(\mathcal{K})$ for any $m\in\nats$,
is used for denoting the operator and the matrix (the operator being
$\bigoplus_{m=1}^\infty C_m$ with $C_m\in B(\mathcal{K})$ for all $m\in\nats$).

In $\mathcal{H}$, consider the unilateral shift operator $S$ and its
adjoint $S^*$ given by
\begin{equation*}
  S  \begin{pmatrix}
    u_1\\ u_2\\ u_3\\\vdots
  \end{pmatrix}=
  \begin{pmatrix}
   0\\ u_1\\ u_2\\\vdots
  \end{pmatrix}
\qquad
  S^*\begin{pmatrix}
    u_1\\ u_2\\ u_3\\\vdots
  \end{pmatrix}=
  \begin{pmatrix}
   u_2\\ u_3\\ u_4\\\vdots
  \end{pmatrix}\,.
\end{equation*}
It can be verified that the operator
\begin{equation}
  \label{eq:representation-j}
  \diag\{B_m\}_{m=1}^\infty + S\diag\{A^*_m\}_{m=1}^\infty +
\diag\{A_m\}_{m=1}^\infty S^*
\end{equation}
coincides with $J$ when it is self-adjoint.


\begin{definition}
  \label{def:green-function}
  Assume that the operator $J$ given in Definition~\ref{def:j-nought}
  is self-adjoint. For any $\zeta$ in the resolvent set of $J$, define
  \begin{equation*}
    G_{jk}(\zeta):=P_j(J-\zeta I)^{-1}P_k\,.
  \end{equation*}
\end{definition}
Note that by this definition $G_{jk}(\zeta)^*=G_{kj}(\cc{\zeta})$,
and therefore
\begin{equation}
  \label{eq:norm-resolvent-adjoint}
  \norm{G_{jk}(\zeta)}_{B(\mathcal{K})}=\norm{G_{kj}(\cc{\zeta})}_{B(\mathcal{K})}\,.
\end{equation}
The operator $G_{jk}(\zeta)$ can be regarded as the entry of the block
matrix representation of the resolvent of $J$ at $\zeta$ corresponding
to the $j$-th row and the $k$-th column. Accordingly, we refer to
$\{G_{jk}(\zeta)\}_{j,k=1}^\infty$ as the block Green matrix
corresponding to $J$ at $\zeta$.

\section{Estimates in a bounded  gap of the essential spectrum}
\label{sec:estimates-bounded-gap-no-commutation}
This section establishes decay bounds for the Green matrix entries
(Definition~\ref{def:green-function})
corresponding to a self-adjoint Jacobi operator $J$
(Definition~\ref{def:j-nought}) with a bounded gap in the essential
spectrum. The real part of the spectral parameter is assumed to be in
this gap and we cover both cases; when this parameter is an eigenvalue
and when it is not.

To streamline the writing of formulae, let us introduce the following
functions. For $x>0$
\begin{equation}
  \label{eq:psi-tilde-psi}
  \psi(x):=x^2 e^{x}\,,\qquad
 \widetilde{\psi}(x):=xe^{x}
\end{equation}
and
\begin{equation}
  \label{eq:phi-delta}
   \phi_\delta(x):=
   \begin{cases}
     1/\delta & \text{ if } 0\le x<\delta\\
     1/x& \text{ if } \delta\le x\,.
   \end{cases}
\end{equation}
Also, let $r,s\in\reals$ be such that $r<s$. For $r<x<s$, define
\begin{equation}
  \label{eq:w-function}
   w(x):=\sqrt{(x-r)(s-x)}>0\,.
 \end{equation}

 The next assertion is part of \cite[Lem.\,3.3]{MR2480099} and turns
 out to be crucial for proving the results of this section.  Earlier
 versions of the statement can be found in \cite{MR1935594}. 
\begin{lemma}
  \label{lem:janas-naboko-stolz}
  Let $T$ be an invertible self-adjoint operator such that
  \begin{align*}
    d_+:&=\dist(0,\sigma(T)\cap(0,+\infty))>0\\
    d_-:&=\dist(0,\sigma(T)\cap(-\infty,0))>0\,.
  \end{align*}
If $V$ is a self-adjoint contraction and $\beta\in\reals$ is such that
\begin{equation}
  \label{eq:sharp-beta}
  \abs{\beta}<\sqrt{d_+d_-}\,,
\end{equation}
then the operator $T+\I\beta V$ is continuously invertible.
\end{lemma}
One can find geometrical heuristics of this result in \cite{MR2480099,MR3028179}.
It is worth mentioning that \eqref{eq:sharp-beta}
cannot be replaced by the weaker condition $\abs{\beta}\le\sqrt{d_+d_-}$
(see \cite[Rem.\,4]{MR3028179}).

To streamline the statements of our results, consider the following
hypothesis.
\begin{hypothesis}
 \label{hyp:main}
  The operator $J$ given in Definition~\ref{def:j-nought}
  is self-adjoint and there is an interval of the real line $(r,s)$ such that
  $(r,s)\cap\sigma_{ess}(J)=\emptyset$.
\end{hypothesis}

For the following assertion, bare in mind the functions given in  \eqref{eq:psi-tilde-psi},
\eqref{eq:phi-delta}, and \eqref{eq:w-function}.

\begin{theorem}
  \label{thm:finite-interval}
  Assume that Hypothesis~\ref{hyp:main} holds true. Fix an arbitrary $\delta>0$,
  $\epsilon\in(0,1/2)$, and $\eta\in(0,1)$.  If
  $\zeta\not\in\sigma(J)$ and $\Re\zeta\in(r,s)$, then
    \begin{equation*}
      \norm{G_{mj}(\zeta)}_{B(\mathcal{K})}\le
    C
   \exp(-\gamma(\zeta)\mkern-12mu\nsum_{k=\min(m,j)}^{\max(m,j)-1}
\mkern-12mu
\phi_\delta\left(\norm{A_k}_{B(\mathcal{K})}\right))\,,
\end{equation*}
where
\begin{equation}
\label{eq:gamma-noncommut-small-imaginary}
  \gamma(\zeta):=\min\left\{\delta\psi^{-1}\left(\frac{w^2(\Re\zeta)\epsilon}{2\delta(s-r)}\right),
\delta\widetilde{\psi}^{-1}\left(\frac{w(\Re\zeta)(1-2\epsilon)}{2\delta}\right)\right\}
\end{equation}
when $\abs{\Im\zeta}\le
w(\Re\zeta)\frac{\epsilon}{2}$, and
\begin{equation}
  \label{eq:gamma-big-imaginary}
  \gamma(\zeta):=\delta\widetilde{\psi}^{-1}\left(\frac{w(\Re\zeta)\epsilon(1-\eta)}{4\delta}\right)
\end{equation}
otherwise. The constant $C$ depends neither on $m$
nor on $j$.
\end{theorem}
\begin{proof}
For any fixed $N\in\nats$, let
\begin{equation}
\label{eq:def-tilde-Phi-m}
  \Phi_m:=
\begin{cases}
\exp\left(-\gamma\msum\limits_{k=1}^{m-1}
\phi_\delta\left(\norm{A_k}_{B(\mathcal{K})}\right)\right)I\,,&
m\le N\,,\\[4mm]
\exp\left(-\gamma\msum\limits_{k=1}^{N-1}
\phi_\delta\left(\norm{A_k}_{B(\mathcal{K})}\right)\right)I\,,&
m> N\,,
\end{cases}
\end{equation}
where $\phi_\delta$ is given in \eqref{eq:phi-delta} and the value of
$\gamma$ is to be determined later. We assume $\Phi_1=I$. Clearly,
$\Phi_m$ is a scalar operator for $m\in\nats$.  Consider the following
operator in $\mathcal{H}$
\begin{equation*}
\Phi:=\diag\{\Phi_m\}_{m=1}^\infty\,.
\end{equation*}
This operator is a contraction and depends on $\delta$, $\gamma$, and $N$. When needed, we
indicate the dependence on $N$ explicitly, i.\,e.,
$\Phi=\Phi(N)$. Note that the operator
$\Phi(N)$ is a boundedly invertible contraction for any
finite $N$.

Let us introduce the operator:
\begin{equation}
\label{eq:relation-tilde-f-phi}
  F:=\Phi^{-1}J\Phi - J\,.
\end{equation}
 Using \eqref{eq:representation-j}, one verifies that
\begin{equation}
  \label{eq:definition-of-f}
  F=S\diag\{\Phi_{m+1}^{-1}A_m^*\Phi_{m}-A_m^*\}
+\diag\{\Phi_m^{-1}A_m\Phi_{m+1}-A_m\}S^*\,.
\end{equation}
By \eqref{eq:def-tilde-Phi-m}, $F$ depends on $N$ and is in
$B(\mathcal{H})$ since the sequences \eqref{eq:sequence} in the range
of $F$ have a finite number of nonzero elements. Note that
\eqref{eq:relation-tilde-f-phi} implies that
\begin{equation}
 \label{eq:tilde-phi-operator-phi}
  \Phi^{-1}(J-\zeta I)\Phi=J+F-\zeta I\,.
\end{equation}
Define
\begin{align}
  X&:=J(E(-\infty,r]+E[s,+\infty)) + rE\left(r,\frac{r+s}{2}\right] +
     sE\left(\frac{r+s}{2},s\right)\,,\label{eq:operatorX}\\
  Y&:=(J-rI)E\left(r,\frac{r+s}{2}\right] +
  (J-sI)E\left(\frac{r+s}{2},s\right)\,,
\end{align}
where $E$ is the spectral measure of $J$. For $M\in\nats$, we also
introduce the following operators. Below, we take $M$ large enough
(see \eqref{eq:first-choice-M-finite-interval},
\eqref{eq:second-choice-M-finite-interval}, \eqref{eq:estimate-b-pm}).
\begin{align}
  \widetilde{J}&:=X+Y(I-\widetilde{P}_M)+F\,, \label{eq:tildeJ}\\
  T&:=X+\Re\left[Y(I-\widetilde{P}_M)+F-\zeta I\right]\,,\\
  R&:=\Im\left[F+Y(I-\widetilde{P}_M)-\zeta I\right]\,,\label{eq:s-definition}
\end{align}
where $\widetilde{P}_M$ is
given in Section~\ref{sec:notation}. Note that $Y$ is a compact operator
and $\widetilde{J}$ depends on $N$ since $F$ also does.

Our aim is to show that $\widetilde{J}-\zeta I$ is continuously
invertible for $\gamma$ indicated in the assertion of the
theorem. Once the inverse of $\widetilde{J}-\zeta I$ exists and is
bounded, the result is established by the next argumentation. Taking
into account that $X+Y=J$ and \eqref{eq:tildeJ}, one obtains from
\eqref{eq:definition-of-f} the equality
\begin{equation*}
  \Phi^{-1}(J-\zeta
  I)\Phi=\widetilde{J}-\zeta I+ Y\widetilde{P}_M\,.
\end{equation*}
Therefore, using the invertibility of $\widetilde{J}-\zeta I$ and the
fact that $\zeta\not\in\sigma(J)$, one concludes that
\begin{equation}
 \label{eq:phi-N-resolvent}
  \Phi^{-1}(J-\zeta
  I)^{-1}\Phi=(\widetilde{J}-\zeta I)^{-1}
\left[I+ Y\widetilde{P}_M(\widetilde{J}-\zeta I)^{-1}\right]^{-1}\,.
\end{equation}
Left multiplying the last equality by $Y\widetilde{P}_M$ yields
\begin{equation*}
  Y\widetilde{P}_M\Phi^{-1}(J-\zeta
  I)^{-1}\Phi=I-
\left[I+ Y\widetilde{P}_M(\widetilde{J}-\zeta I)^{-1}\right]^{-1}\,,
\end{equation*}
where we have used the identity $Q(I+Q)^{-1}=I-(I+Q)^{-1}$ formally valid for
any linear operator $Q$.
Inserting the expression for $\left[I+
  Y\widetilde{P}_M(\widetilde{J}-\zeta I)^{-1}\right]^{-1}$ from the
last equality into \eqref{eq:phi-N-resolvent}, one arrives at
\begin{equation*}
  \Phi^{-1}(J-\zeta
  I)^{-1}\Phi=(\widetilde{J}-\zeta I)^{-1}
\left[I-Y\widetilde{P}_M\Phi^{-1}(J-\zeta
  I)^{-1}\Phi\right]\,.
\end{equation*}
In the right-hand side of the previous equality, the inverse of the
operator in square brackets in \eqref{eq:phi-N-resolvent} has been
substituted by an explicit expression.  The rigorous proof of the
existence of this inverse follows by checking directly that the
obtained expression is the left and right inverse simultaneously.

Clearly,
\begin{align*}
  &\norm{I-Y\widetilde{P}_M\Phi^{-1}(J-\zeta
  I)^{-1}\Phi}_{B(\mathcal{H})}\\ &\le
 1+\norm{Y}_{B(\mathcal{H})}\norm{\widetilde{P}_M\Phi^{-1}(N)}_{B(\mathcal{H})}
\norm{(J-\zeta
  I)^{-1}}_{B(\mathcal{H})}\norm{\Phi(N)}\\
&\le 1+\norm{Y}_{B(\mathcal{H})}\norm{\widetilde{P}_M\Phi^{-1}(M)}_{B(\mathcal{H})}
\norm{(J-\zeta
  I)^{-1}}_{B(\mathcal{H})}\\
&:=\widetilde{C}\,.
\end{align*}
Note that $\widetilde{C}$ does not depend on $N$.
Thus,
\begin{equation}
 \label{eq:bounded-resolvent-indep-N}
  \norm{P_j\Phi^{-1}(N)(J-\zeta
  I)^{-1}\Phi(N)P_k}_{B(\mathcal{H})}\le
\widetilde{C}
\norm{(\widetilde{J}-\zeta I)^{-1}}_{B(\mathcal{H})}\,.
\end{equation}
In view of
Definition~\ref{def:green-function} and \eqref{eq:def-tilde-Phi-m},
by letting $N\to\infty$, one obtains
\begin{equation}
 \label{eq:final-inequality-first}
  \norm{\exp\left(\gamma\sum_{m=1}^{j-1}\phi_\delta(\norm{A_m}_{B(\mathcal{K})})\right)
G_{jk}(\zeta)
\exp\left(-\gamma\sum_{m=1}^{k-1}\phi_\delta(\norm{A_m}_{B(\mathcal{K})})\right)}_{B(\mathcal{K})}
\le C
\end{equation}
for all $j,k\in\nats$. The estimate of the theorem is proven by
combining both scalar exponential factors in
\eqref{eq:final-inequality-first}. Note that we have proven the
estimate for $j\ge k$, but the other case is also covered by recurring
to \eqref{eq:norm-resolvent-adjoint}.

Let us turn to establishing the continuous invertibility of
$\widetilde{J}-\zeta I$. We consider two mutually exclusive cases
leading to the two different definitions of the function $\gamma$. If
$\abs{\Im\zeta}\le w(\Re\zeta)\frac{\epsilon}{2}$, then one
establishes the invertibility of $\widetilde{J}-\zeta I$ on the basis
of Lemma~\ref{lem:janas-naboko-stolz} by choosing $\gamma$
appropriately. If $\abs{\Im\zeta}> w(\Re\zeta)\frac{\epsilon}{2}$,
then one can rely on the properties of dissipative (or
anti-dissipative) operators to conclude that $\widetilde{J}-\zeta I$ is
invertible for some values of $\gamma$.

First we assume that $\abs{\Im\zeta}\le
w(\Re\zeta)\frac{\epsilon}{2}$.
Observe that we have defined $T$ and $R$  so that
\begin{equation*}
  \widetilde{J}-\zeta I= T+\I R\,.
\end{equation*}
Let us verify that the operators $T$ and $R$ satisfy the conditions of
Lemma~\ref{lem:janas-naboko-stolz}. To this end, we estimate the norms
of the real and imaginary parts of $F$.

It follows from ~\eqref{eq:definition-of-f} that
\begin{align*}
  2\Re F&=F+F^*\\ &=
S\diag\{\Phi_m^{-1}A_m\Phi_{m+1}-2A_m+\Phi_m^*A_m(\Phi_{m+1}^*)^{-1}\}^*\\
 &+\diag\{\Phi_m^{-1}A_m\Phi_{m+1}-2A_m+\Phi_m^*A_m(\Phi_{m+1}^*)^{-1}\}S^*\,.
\end{align*}
Since the matrix of the operator $2\Re F$ has only two block diagonals
not necessarily zero and one diagonal is the adjoint of the other, one
has
\begin{equation}
  \label{eq:estimate-real-F}
\norm{\Re F}_{B(\mathcal{H})}\le\sup_{m\in\nats}\left\{\norm{\Phi_m^{-1}A_m\Phi_{m+1}-2A_m+
\Phi_m^*A_m(\Phi_{m+1}^*)^{-1}}_{B(\mathcal{K})}\right\}\,,
\end{equation}
which in view of \eqref{eq:def-tilde-Phi-m} yields
\begin{equation*}
    \norm{\Re F}_{B(\mathcal{H})}\le
\sup_{m\in\nats}
\left\{\norm{A_m\left(e^{-\gamma\phi_{\delta}(\norm{A_m}_{B(\mathcal{K})})}
-2+e^{\gamma\phi_\delta(\norm{A_m}_{B(\mathcal{K})})}\right)}_{B(\mathcal{K})}\right\}\,.
\end{equation*}
Using the elementary inequality
\begin{equation}
 \label{eq:algebraic-inequality-exp}
  0\le e^x-2+e^{-x}\le x^2e^x\,,
\end{equation}
which holds for any $x\ge 0$, one arrives at
\begin{equation}
  \label{eq:norm-real-part-supremum}
    \norm{\Re F}_{B(\mathcal{H})}\le
\sup_{m\in\nats}
\left\{\norm{A_m}_{B(\mathcal{K})}\gamma^2\phi_\delta^2(\norm{A_m}_{B(\mathcal{K})})
e^{\gamma\phi_\delta(\norm{A_m}_{B(\mathcal{K})})}\right\}\,.
\end{equation}
To obtain an appropriate estimate for $\norm{\Re F}$, first note that, for any given
$\delta,\epsilon>0$, if one defines
\begin{equation}
  \label{eq:gamma-first-definition}
  \gamma:=\delta\psi^{-1}
\left(\frac{w^2(\Re\zeta)\epsilon}{2\delta(s-r)}\right)\,,
\end{equation}
then
\begin{equation}
  \label{eq:inequality-to-hold}
  \xi\psi(\gamma\phi_\delta(\xi))\le
  \frac{w^2(\Re\zeta)\epsilon}{2(s-r)}
\end{equation}
for all $\xi>0$. Indeed, it follows from \eqref{eq:phi-delta} that,
when $0<\xi\le\delta$, the inequality
\eqref{eq:inequality-to-hold} holds whenever
\begin{equation*}
  \psi\left(\frac\gamma\delta\right)\le\frac{w^2(\Re \zeta)\epsilon}{2\delta(s-r)}\,.
\end{equation*}
Hence, \eqref{eq:gamma-first-definition} guarantees
\eqref{eq:inequality-to-hold}. In the case $\xi>\delta$, assuming
\eqref{eq:gamma-first-definition}, rewrite
\eqref{eq:inequality-to-hold} using \eqref{eq:phi-delta} as follows
\begin{equation*}
  \xi\psi\left(\frac\gamma\xi\right)\le\psi\left(\frac\gamma\delta\right)\delta\,.
\end{equation*}
By \eqref{eq:psi-tilde-psi}, this inequality is
equivalent to
\begin{equation*}
  \frac{\gamma^2}\xi e^{\frac\gamma\xi}\le\frac{\gamma^2}\delta e^{\frac\gamma\delta}\,,
\end{equation*}
so, after dividing by $\gamma$, one verifies that it holds since
$\xi>\delta$ and the function $te^t$ is monotone for $t>0$.

Comparing \eqref{eq:norm-real-part-supremum} to
\eqref{eq:inequality-to-hold}, one concludes that
\begin{equation}
 \label{eq:estimate-real-part-finite-interval}
  \norm{\Re F}_{B(\mathcal{H})}\le
\frac{w^2(\Re\zeta)\epsilon}{2(s-r)}\le
\frac{\dist(\Re\zeta,
    \{r,s\})\epsilon}{2}\,.
\end{equation}
Since $Y$ is compact, one can take
$M$ sufficiently large so that, simultaneously,
\begin{equation}
 \label{eq:first-choice-M-finite-interval}
  \norm{\Re Y(I-\widetilde{P}_M)}_{B(\mathcal{H})}\le
(\Re\zeta-r)\frac{\epsilon}{2}\,,\quad
\norm{\Re Y(I-\widetilde{P}_M)}_{B(\mathcal{H})}\le
  (s-\Re\zeta)\frac{\epsilon}{2}\,.
\end{equation}
Thus, taking into account that
\begin{align*}
  d_-&\ge (\Re\zeta-r)-\norm{\Re
    Y(I-\widetilde{P}_M)}_{B(\mathcal{H})}
-\norm{\Re F}_{B(\mathcal{H})}\\
  d_+&\ge (s-\Re\zeta)-\norm{\Re
    Y(I-\widetilde{P}_M)}_{B(\mathcal{H})}
-\norm{\Re F}_{B(\mathcal{H})}\,,
\end{align*}
and combining \eqref{eq:estimate-real-part-finite-interval} and
\eqref{eq:first-choice-M-finite-interval}, one obtains
\begin{equation}
\label{eq:bound-for-d+-}
\begin{split}
  d_-&\ge (\Re\zeta-r)(1-\epsilon)\\
  d_+&\ge (s-\Re\zeta)(1-\epsilon)\,.
\end{split}
\end{equation}

On the other hand, by \eqref{eq:s-definition}, one has
\begin{equation}
\label{eq:first-estimate-for-beta}
  \norm{R}_{B(\mathcal{H})}
\le\norm{\Im Y(I-\widetilde{P}_M)}_{B(\mathcal{H})}
+\norm{\Im F}_{B(\mathcal{H})}+\abs{\Im\zeta}\,.
\end{equation}
Make $M$ perhaps even larger than what was required in
\eqref{eq:first-choice-M-finite-interval} so that
\begin{equation}
 \label{eq:second-choice-M-finite-interval}
  \norm{\Im Y(I-\widetilde{P}_M)}_{B(\mathcal{H})}\le w(\Re\zeta)\frac{\epsilon}{4}.
\end{equation}
To found a bound for the second term in the right-hand side of
\eqref{eq:first-estimate-for-beta}, one uses the argumentation for
obtaining the inequality \eqref{eq:estimate-real-F} from the
expression \eqref{eq:definition-of-f}. Thus,
\begin{equation}
  \label{eq:estimate-imaginary-F}
\norm{\Im F}_{B(\mathcal{H})}\le\sup_{m\in\nats}\left\{\norm{\Phi_m^{-1}A_m\Phi_{m+1}-
\Phi_m^*A_m(\Phi_{m+1}^*)^{-1}}_{B(\mathcal{K})}\right\}\,.
\end{equation}
A straightforward computation yields
\begin{equation}
  \label{eq:F-imaginary-part-first-estimate}
  \norm{\Im F}_{B(\mathcal{H})}\le
\sup_{m\in\nats}
\left\{\norm{A_m\left(e^{-\gamma\phi_{\delta}(\norm{A_m}_{B(\mathcal{K})})}
-e^{\gamma\phi_\delta(\norm{A_m}_{B(\mathcal{K})})}\right)}_{B(\mathcal{K})}\right\}\,,
\end{equation}
from which, by means of the elementary inequality
\begin{equation*}
0\le e^x-e^{-x}\le 2xe^x\,,
\end{equation*}
valid for any positive $x$, one obtains
\begin{equation}
  \label{eq:norm-F-imaginary-part}
  \norm{\Im F}_{B(\mathcal{H})}\le\sup_{m\in\nats}
\left\{2\norm{A_m}_{B(\mathcal{K})}\gamma
\phi_\delta\left(\norm{A_m}_{B(\mathcal{K})}\right)
e^{\gamma\phi_\delta(\norm{A_m}_{B(\mathcal{K})})}\right\}\,.
\end{equation}
Proceeding as before, if one defines
\begin{equation}
  \label{eq:gamma-second-definition}
  \gamma:=\delta\widetilde{\psi}^{-1}
\left(\frac{w(\Re\zeta)(1-2\epsilon)}{2\delta}\right)\,,
\end{equation}
assuming $\delta>0$ and $\epsilon\in(0,1/2)$, then
\begin{equation*}
  2\xi\widetilde{\psi}(\gamma\phi_\delta(\xi))\le
  w(\Re\zeta)(1-2\epsilon)
\end{equation*}
for all $\xi>0$. Comparing this inequality with
\eqref{eq:norm-F-imaginary-part}, one concludes that
\begin{equation}
 \label{eq:estimate-imaginary-part-finite-interval}
  \norm{\Im F}_{B(\mathcal{H})}\le w(\Re\zeta)(1-2\epsilon)\,.
\end{equation}
Taking into account that $\abs{\Im\zeta}\le
w(\Re\zeta)\frac\epsilon2$ and inserting
\eqref{eq:second-choice-M-finite-interval} and
\eqref{eq:estimate-imaginary-part-finite-interval} into
\eqref{eq:first-estimate-for-beta}, one arrives at
\begin{equation}
  \label{eq:final-estimate-beta}
  \norm{R}_{B(\mathcal{H})}\le w(\Re\zeta)(1-\tfrac{5}{4}\epsilon)<
  w(\Re\zeta)(1-\epsilon)\le\sqrt{d_-d_+}\,,
\end{equation}
where the second inequality is obtained from \eqref{eq:bound-for-d+-}.

The estimates \eqref{eq:estimate-real-part-finite-interval} and
\eqref{eq:first-choice-M-finite-interval} show that $T$ is
invertible. Clearly, $V:=R/\norm{R}$ is a self-adjoint contraction. Thus, since according
to \eqref{eq:final-estimate-beta} $\beta:=\norm{R}$ satisfies the conditions of
Lemma~\ref{lem:janas-naboko-stolz}, one concludes that
$\widetilde{J}-\zeta I$ is invertible.

Let us now consider the case
\begin{equation}
  \label{eq:case-imaginary-lambda-big}
\abs{\Im\zeta}>
w(\Re\zeta)\frac\epsilon2\,.
\end{equation}
As before, we first choose $M\in\nats$
sufficiently large so that
\begin{equation}
  \label{eq:estimate-b-pm}
  \norm{Y(I-P_M)}_{B(\mathcal{H})}\le w(\Re\zeta)\frac{\epsilon\eta}{4}\,,\qquad 0<\eta<1\,.
\end{equation}
Now, if $\gamma$ is as in \eqref{eq:gamma-big-imaginary}, then
\begin{equation}
  \label{eq:xi-inequality-imaginary-big}
    2\xi\widetilde{\psi}(\gamma\phi_\delta(\xi))\le
  w(\Re\zeta)\frac\epsilon2(1-\eta)\,.
\end{equation}
This is shown following the reasoning used to established that
\eqref{eq:inequality-to-hold} holds under the assumption
\eqref{eq:gamma-first-definition}. Comparing
\eqref{eq:xi-inequality-imaginary-big} and
\eqref{eq:norm-F-imaginary-part}, one obtains
\begin{equation}
  \label{eq:norm-imaginary-F-big}
   \norm{\Im F}_{B(\mathcal{H})}\le w(\Re\zeta)\frac\epsilon2(1-\eta)\,.
 \end{equation}
Since $\Im (\widetilde{J}-\zeta I)= R$, it follows from \eqref{eq:s-definition},
\eqref{eq:estimate-b-pm} and \eqref{eq:norm-imaginary-F-big} that
\begin{equation*}
  \Im (\widetilde{J}-\zeta I)\le w(\Re\zeta)\frac\epsilon2(1-\eta)
  +w(\Re\zeta)\frac{\epsilon\eta}{4} -\Im\zeta I\,.
\end{equation*}
Thus, if one considers the particular realization of
\eqref{eq:case-imaginary-lambda-big} given by $\Im\zeta>
w(\Re\zeta)\frac\epsilon2$, then
\begin{align*}
  \Im (\widetilde{J}-\zeta I)&< w(\Re\zeta)\frac\epsilon2(1-\eta)
  +w(\Re\zeta)\frac{\epsilon\eta}{4} - w(\Re\zeta)\epsilon/2\\
  &= -w(\Re\zeta)\frac{\epsilon\eta}{4}\,.
\end{align*}
Thus the operator $\widetilde{J}-\zeta I
+iw(\Re\zeta)\frac{\epsilon\eta}{4}$ is anti-dissipative and
therefore $\widetilde{J}-\zeta I$ is invertible. The case $\Im\zeta<-
w(\Re\zeta)\frac\epsilon2$ is treated analogously.
\end{proof}
\begin{remark}
 \label{rem:delta-parameter}
  In the estimate given by Theorem~\ref{thm:finite-interval}, there
  are no constraints in the choice of the constant $\delta>0$. The
  function $\gamma(\zeta)$ yields \emph{a priori} better estimates when
  $\delta\to\infty$. Note, however, that one does not obtain a more
  accurate estimate by letting $\delta\to\infty$ due to the dependence
  on $\delta$ of the function $\phi_\delta$.
\end{remark}
\begin{remark}
  \label{rem:infinite-fnite-gap}
  The hypothesis of Theorem~\ref{thm:finite-interval} is fulfilled by
  a bounded from below self-adjoint operator $J$ with bounded from
  below essential spectrum. For these operators
  \cite[Thm.\,3.1]{MR3820400} provides an estimate for the decay of
  Green matrix entries which is more precise than the one given by
  Theorem~\ref{thm:finite-interval}. Note that that boundedness from
  below is not required by Theorem~\ref{thm:finite-interval}.  In
  \cite[Thm.\,3.1]{MR3820400}, this requirement was crucial for the
  proof. As shown in \cite{MR2579689}, the asymptotic behaviour of
  generalized eigenvectors of a non-bounded from below Jacobi operator
  can depend on the main diagonal in contrast to
  \cite[Thm.\,3.1]{MR3820400} and Theorem~\ref{thm:finite-interval}.
  The results of \cite{MR2579689} illustrate why the semiboundedness
  is essential in the case of \cite[Thm.\,3.1]{MR3820400} and does not
  contradicts Theorem~\ref{thm:finite-interval} due to the different
  order of estimates.
\end{remark}

For the next theorem, we rely on
\cite[Lem.\,3.1]{MR3820400}. Here we reproduce the statement of the lemma for
easy reference. Recall that $P_1$ is given in Section~\ref{sec:notation}.

\begin{lemma}
  \label{lem:perturbed-not-in-spectrum}
  Let $J$ be the operator given in Definition~\ref{def:j-nought} and
  $L$ be a compact operator in $\mathcal{K}$ with trivial kernel such that
  $\norm{L}_{B(\mathcal{K})}=1$. If $A_m$ has trivial kernel for all
  $m\in\nats$ and $\zeta$ is in the discrete spectrum of $J$, then,
  for any  $\tau>0$ sufficiently small, $\zeta$ is not in the spectrum of
  \begin{equation}
    \label{eq:t-epsilon}
    J(\tau):=J+\tau P_1L^*LP_1\,.
  \end{equation}
\end{lemma}

\begin{theorem}
  \label{thm:finite-interval-in-spectrum}
  Assume that Hypothesis~\ref{hyp:main} holds true and that
  $\ker(A_m)=\{0\}$ for any $m\in\nats$. Fix an arbitrary $\delta>0$
  and $\epsilon\in(0,1/2)$.
 If $\zeta\in\sigma_p(J)\cap (r,s)$ and $u$ is the corresponding
  eigenvector, then
  \begin{equation*}
    \norm{u_m}_{\mathcal{K}}\le
C\exp(-\gamma(\zeta)\nsum_{k=1}^{m-1}
\phi_\delta\left(\norm{A_k}_{B(\mathcal{K})}\right))\,,
  \end{equation*}
where $\gamma$ and $\phi_{\delta}$ are given by
\eqref{eq:gamma-noncommut-small-imaginary} and \eqref{eq:phi-delta}, respectively.
The  constant $C$ does not depend on $m$.
\end{theorem}
\begin{proof}
  By Lemma~\ref{lem:perturbed-not-in-spectrum}, one can choose
  $\tau>0$ such that $\zeta\not\in\sigma(J(\tau))$. Additionally,
  according to perturbation theory,
  $\sigma_{ess}(J)=\sigma_{ess}(J(\tau))$ (see \cite[Chap.\,4,
  Thm.\,5.35 and Chap.\,5, Thm.\,4.11]{MR0407617}). Therefore $J(\tau)$
  satisfies the hypothesis of Theorem~\ref{thm:finite-interval}.

  Now, if $u$ is a nonzero vector in $\ker(J-\zeta I)$, then
  \begin{equation*}
    (J(\tau)-\zeta I)u=\tau P_1L^*LP_1u\,.
  \end{equation*}
Therefore
\begin{equation}
\label{eq:eigenvector-throu-itself}
   u=\tau(J(\tau)-\zeta I)^{-1}P_1L^*LP_1u\,,
\end{equation}
which in turn implies
\begin{align*}
  \norm{u_m}_{\mathcal{K}}&=\norm{P_mu}\\ &=\norm{P_m\tau(J(\tau)-\zeta
                             I)^{-1}P_1L^*LP_1u}\\
&\le\tau\norm{P_m(J(\tau)-\zeta I)^{-1}P_1}_{B(\mathcal{H})}
\norm{P_1u}\\
&\le \widetilde{C}\exp(-\gamma(\zeta)\msum_{k=1}^{m-1}
\phi_\delta\left(\norm{A_k}_{B(\mathcal{K})}\right))\norm{P_1u}\,,
\end{align*}
where in the first inequality we use that
$\norm{L}_{B(\mathcal{K})}=1$ and in the second we resort to
Theorem~\ref{thm:finite-interval}.
\end{proof}

\begin{corollary}
  \label{cor:norm-to-infty-finite}
  Assume that Hypothesis~\ref{hyp:main} holds and that
  $\norm{A_k}_{B(\mathcal{K})}\!\convergesto{k}\infty$.
  Fix $\epsilon\in(0,1/2)$.
\begin{enumerate}[(a)]
\item \label{not-in-spectrum-cor-finite}
If $\zeta\not\in\sigma(J)$ and $\Re\zeta\in(r,s)$, then
  \begin{equation*}
       \norm{G_{mj}(\zeta)}_{B(\mathcal{K})}\le
    C_a\exp(-\widetilde{\gamma}(\zeta)\mkern-12mu
  \nsum_{k=\min(m,j)}^{\max(m,j)-1}\mkern-12mu\norm{A_k}_{B(\mathcal{K})}^{-1})\,,
  \end{equation*}
where
\begin{equation}
\label{eq:gamma-tilde}
  \widetilde{\gamma}(\zeta)<
  \begin{cases}
    w(\Re\zeta)\left(\frac12-\epsilon\right) & \text{ if }\mkern6mu
 \abs{\Im(\zeta)}\le
w(\Re\zeta)\frac{\epsilon}{2}\\[3mm]
\frac{w(\Re\zeta)}{4}\epsilon & \text{ otherwise}.
  \end{cases}
\end{equation}
\item \label{in-spectrum-cor-finite} If $\ker(A_n)$ is trivial for all
$n\in\nats$, $\zeta\in\sigma_p(J)\cap (r,s)$ and $u$ is the corresponding
eigenvector,  then
  \begin{equation*}
    \norm{u_m}_{\mathcal{K}}\le
   C_b\exp(-w(\zeta)\left(\frac12-\epsilon\right)\nsum_{k=1}^{m-1}
\norm{A_k}_{B(\mathcal{K})}^{-1})\,.
  \end{equation*}
  \end{enumerate}
The constant $C_a$ does not
depend on $m$ and $j$, and $C_b$ does not depend on $m$.
\end{corollary}
\begin{remark}
  \label{rem:results-corollary}
  Perhaps the most relevant case in various theoretical applications
  corresponds to $\zeta$ being actually in the gap of the essential
  spectrum in item (a). This case admits further simplification. Indeed,
  it follows from \eqref{not-in-spectrum-cor-finite} that if
  $\zeta\in(r,s)$, then, for arbitrarily small $\epsilon'\in(0,1/2)$,
  one has
  \begin{equation*}
  \norm{G_{mj}(\zeta)}_{B(\mathcal{K})}\le
    C_a\exp(-w(\zeta)\left(\frac12-\epsilon'\right)\mkern-12mu
  \nsum_{k=\min(m,j)}^{\max(m,j)-1}\mkern-12mu\norm{A_k}_{B(\mathcal{K})}^{-1})\,.
  \end{equation*}
\end{remark}
\begin{proof}
  First note that the expressions inside the minimum in
  \eqref{eq:gamma-noncommut-small-imaginary} monotonically grow as
  $\delta\to\infty$. Since the minimum has to be taken, one should
  consider only the expression  that grows slower, namely,
  \begin{equation*}
    \delta\widetilde{\psi}^{-1}
\left(\frac{w(\Re\zeta)(1-2\epsilon)}{2\delta}\right)\,.
  \end{equation*}
By choosing $\delta$ appropriately (essentially sufficiently large),
  one obtains
  \begin{equation*}
    \frac{w(\Re\zeta)(1-2\epsilon)}{2\delta}<\epsilon_1\ll 1\,.
  \end{equation*}
  Given $\epsilon$ and $\epsilon_1$, the choice of $\delta$ depends on
  $\zeta$.

If $0<t<\epsilon_1$, then
\begin{equation}
 \label{eq:correct-inequality-psi-}
  \psi^{-1}(t)\ge t(1-\eta)\,,
\end{equation}
where $\eta$ is arbitrarily small whenever $\epsilon_1$
is sufficiently small. Indeed, it follows from
\eqref{eq:psi-tilde-psi} that if $t=y\exp(y)$, then
\begin{equation*}
  \widetilde{\psi}^{-1}(t)=t\exp(-y)=
  t\exp(-t\exp(-t))
\ge t\exp(-t)\,.
\end{equation*}
Thus, the assumption $t\le\epsilon_1$ implies
\begin{equation*}
  \widetilde{\psi}^{-1}(t)\ge t\exp(\epsilon_1)(1-\eta)\,.
\end{equation*}
In view of
\eqref{eq:correct-inequality-psi-}, the choice of $\gamma$
in \eqref{eq:gamma-noncommut-small-imaginary} can be replaced by
\begin{equation*}
  \gamma=\delta\frac{w(\Re\zeta)(1-2\epsilon)}{2\delta}(1-\eta)\,.
\end{equation*}
Finally, observe that, for any $p\in\nats$,
\begin{equation}
\label{eq:argumentation-cor}
\begin{split}
  \msum_{k=1}^p\phi_\delta\left(\norm{A_k}_{B(\mathcal{K})}\right)&=
 \msum_{\substack{\left\| A_k\right\|\le\delta\\ k
\le p}}\frac{1}{\delta}+
\msum_{\substack{\left\| A_k\right\|>\delta\\ k
\le p}}\frac{1}{\norm{A_k}_{B(\mathcal{K})}}\\
&=\msum_{k=1}^p \frac{1}{\norm{A_k}_{B(\mathcal{K})}} +
 \msum_{\substack{\left\| A_k\right\|\le\delta\\ k
\le p}}\frac{1}{\delta}-
\msum_{\substack{\left\| A_k\right\|\le\delta\\ k
\le p}}\frac{1}{\norm{A_k}_{B(\mathcal{K})}}\,,
\end{split}
\end{equation}
where the second and third terms can be absorbed into a constant
which does not depend on $p\gg 1$ since $\norm{A_k}_{B(\mathcal{K})}\convergesto{k}\infty$.
\end{proof}

\section{Estimates in the case of commuting entries}
\label{sec:green-funct-estim}
The results of the previous section admit a refinement when  the
hypotheses of Theorems~\ref{thm:finite-interval} and
\ref{thm:finite-interval-in-spectrum} are complemented with the
requirement that $A_m,B_m,A_m^*$ commute for
$m\in\nats$. This refinement permits to have different bound along
different vectors in the space $\mathcal{K}$.

\begin{hypothesis}
 \label{hyp:commute}
  The system of operators
      $\{A_m,B_m,A^*_m\}_{m\in\nats}$
    given in section \ref{sec:block-jacobi-matr} commutes pairwise.
\end{hypothesis}

\begin{theorem}
  \label{thm:finite-interval-commutation}
Assume that Hypotheses~\ref{hyp:main} and \ref{hyp:commute} hold true.
 Fix an arbitrary $\delta>0$,
  $\epsilon\in(0,1/2)$, and $\eta\in(0,1)$. If $\zeta\not\in\sigma(J)$
  and $\Re\zeta\in(r,s)$, then
    \begin{equation*}
 \norm{\exp(\gamma(\zeta)\mkern-12mu\nsum\limits_{k=\min(m,j)}^{\max(m,j)-1}
\mkern-12mu\phi_\delta(\abs{A_k}))G_{mj}(\zeta)}_{B(\mathcal{K})}
\le C\,,
\end{equation*}
where $\gamma$ is given by \eqref{eq:gamma-noncommut-small-imaginary}
if $\abs{\Im(\zeta)}\le
w(\Re\zeta)\frac{\epsilon}{2}$, and by
\eqref{eq:gamma-big-imaginary} otherwise.
The constant $C$ does not depend on $m$ and $j$.
\end{theorem}
\begin{proof}
Consider the operators given in
\eqref{eq:operatorX}--\eqref{eq:s-definition}. Now, we modify the definition of the
  operator $\Phi$. For any fixed $N\in\nats$, redefine the bounded
  operators on $\mathcal{K}$ given in the proof of Theorem~\ref{thm:finite-interval}:
  \begin{equation}
    \label{eq:phi-m-new-def}
      \Phi_m:=
\begin{cases}
\exp\left(-\gamma\msum\limits_{k=1}^{m-1}\phi_\delta(\abs{A_k})\right)\,,&
m\le N\,,\\[4mm]
\exp\left(-\gamma\msum\limits_{k=1}^{N-1}\phi_\delta(\abs{A_k})\right)\,,&
m> N\,.
\end{cases}
  \end{equation}
The bounded operator $\Phi$ on $\mathcal{H}$ is defined by
\begin{equation*}
  \Phi:=\diag\{\Phi_m\}_{m=1}^\infty\,.
\end{equation*}
Note that this operator differs from its counterpart of the proof of
Theorem~\ref{thm:finite-interval}.
Similar to what we had in the proof of
Theorem~\ref{thm:finite-interval}, $\Phi$ depends on $N$ and
$\Phi(N)$ is a boundedly invertible contraction for any finite
$N$. Note that this time the block operator $\Phi_m$ is not a scalar
operator.

Define the operator $F$ by \eqref{eq:definition-of-f} with the new
sequence $\{\Phi_m\}_{m\in\nats}$. Repeating
the argumentation in the proof of Theorem~\ref{thm:finite-interval},
one arrives at \eqref{eq:estimate-real-F}. Using
\eqref{eq:phi-m-new-def} and the fact that the system
$\{A_m,B_m,A^*_m\}_{m\in\nats}$ commutes pairwise, one obtains from
\eqref{eq:estimate-real-F} that
\begin{equation}
 \label{eq:inequality-real-part-modulus-operator}
    \norm{\Re F}_{B(\mathcal{H})}\le
\sup_{m\in\nats}
\left\{\norm{\abs{A_m}\left(e^{-\gamma\phi_{\delta}(\abs{A_m})}
-2I+e^{\gamma\phi_\delta(\abs{A_m})}\right)}_{B(\mathcal{K})}\right\}\,.
\end{equation}
Due to the inequality
\begin{equation*}
  e^Q-2 I+e^{-Q}\le Q^2e^Q
\end{equation*}
valid for any positive operator $Q$ and obtained from
\eqref{eq:algebraic-inequality-exp} by the spectral theorem, one derives
from \eqref{eq:inequality-real-part-modulus-operator} the estimate
\begin{equation*}
   \norm{\Re F}_{B(\mathcal{H})}\le
\sup_{m\in\nats}
\left\{\norm{\abs{A_m}\gamma^2\phi_\delta^2(\abs{A_m})
e^{\gamma\phi_\delta(\abs{A_m})}}\right\}\,.
\end{equation*}
Taking $\gamma$ as in \eqref{eq:gamma-first-definition}, one concludes
from \eqref{eq:inequality-to-hold} and the spectral theorem, that
\eqref{eq:estimate-real-part-finite-interval} holds.

Similarly, it follows from \eqref{eq:F-imaginary-part-first-estimate}
that
\begin{equation*}
    \norm{\Im F}_{B(\mathcal{H})}\le
\sup_{m\in\nats}
\left\{\norm{\abs{A_m}\left(e^{-\gamma\phi_{\delta}(\abs{A_m}_{B(\mathcal{K})})}
-e^{\gamma\phi_\delta(\abs{A_m}_{B(\mathcal{K})})}\right)}_{B(\mathcal{K})}\right\}\,.
\end{equation*}
This inequality implies, by means of the operator inequality
\begin{equation*}
0\le e^Q-e^{-Q}\le 2Qe^Q\,,
\end{equation*}
which holds for any positive operator due to the spectral theorem,
that
\begin{equation*}
\norm{\Im F}_{B(\mathcal{H})}\le\sup_{m\in\nats}
\left\{\norm{\abs{A_m}\gamma
    \phi_\delta(\abs{A_m}e^{\gamma\phi_\delta(\abs{A_m})}}_{B(\mathcal{K})}\right\}\,.
\end{equation*}
Thus, by choosing $\gamma$ as in \eqref{eq:gamma-second-definition},
one verifies through the spectral theorem that
\eqref{eq:estimate-imaginary-part-finite-interval} holds.

Since \eqref{eq:estimate-real-part-finite-interval} and
\eqref{eq:estimate-imaginary-part-finite-interval} take place, the
operator $\widetilde{J}-\zeta I$ is invertible and therefore one has
the estimate \eqref{eq:bounded-resolvent-indep-N}. Thus, in view of
Definition~\ref{def:green-function} and \eqref{eq:phi-m-new-def}, if
$N\to\infty$, then
\begin{equation*}
    \norm{\exp\left(\gamma\msum_{m=1}^{j-1}\phi_\delta(\abs{A_m})\right)
G_{jk}(\zeta)
\exp\left(-\gamma\msum_{m=1}^{k-1}\phi_\delta(\abs{A_m})\right)}_{B(\mathcal{K})}
\le C
\end{equation*}
for all $j,k\in\nats$. The assertion of the theorem follows from this
inequality by combining the operators on both sides of $G_{jk}$. In
this proof, $j\ge k$, but the other case is also covered by recurring
to \eqref{eq:norm-resolvent-adjoint}.
\end{proof}
\begin{theorem}
  \label{thm:finite-interval-commutation-in-spectrum}
Assume that Hypotheses~\ref{hyp:main} and \ref{hyp:commute} hold true
and that $\ker(A_m)$ is trivial for any $m\in\nats$.
Fix an arbitrary $\delta>0$
  and $\epsilon\in(0,1/2)$.
 If $\zeta\in\sigma_p(J)\cap (r,s)$ and $u$ is the corresponding
  eigenvector, then
  \begin{equation*}
    \norm{\exp(\gamma(\zeta)\sum_{k=1}^{m-1}
\phi_\delta(\abs{A_k}))u_m}_{\mathcal{K}}\le C\,,
  \end{equation*}
where $\gamma$ is given by \eqref{eq:gamma-noncommut-small-imaginary}.
The constant $C$ does not depend on $m$.
\end{theorem}
\begin{proof}
  One follows the argumentation of the proof of
  Theorem~\ref{thm:finite-interval-in-spectrum}. Resort to
  Lemma~\ref{lem:perturbed-not-in-spectrum} and choose $\tau>0$ so
  that $\zeta\not\in\sigma(J(\tau))$.
 It follows from
  \eqref{eq:t-epsilon} that if $u$ is in $\ker(J-\zeta I)$, then
  \eqref{eq:eigenvector-throu-itself} holds. Thus,
  \begin{equation*}
    P_m\Phi^{-1} u=\tau P_m\Phi^{-1}(J(\tau)-\zeta I)^{-1}P_1L^*LP_1u\,.
  \end{equation*}
 One then obtains from this expression that
 \begin{equation*}
    \norm{\Phi^{-1}(N)P_m u}\le
\tau\norm{P_m\Phi^{-1}(m)(J(\tau)-\zeta I)^{-1}P_1L^*LP_1u}\,.
 \end{equation*}
 For finishing the proof, it only remains to let $N\to\infty$ and note
 that $J(\tau)-\zeta I$ is continuously invertible.
\end{proof}

\begin{corollary}
  \label{cor:norm-to-infty-commutation-finite}
  Fix $\epsilon\in(0,1/2)$. Assume that Hypotheses~\ref{hyp:main} and
  \ref{hyp:commute} hold true and
  $\norm{A_m^{-1}}_{B(\mathcal{K})}\mkern-1mu\convergesto{m}0$
  ($A_{m}$ invertible for $m\ge m_{0}\in\nats$).
\begin{enumerate}[a)]
\item \label{not-in-spectrum-cor-commutation-finite}
If $\zeta\not\in\sigma(J)$ and $\Re\zeta\in(r,s)$, then
  \begin{equation*}
       \norm{\exp( \widetilde{\gamma}(\zeta)\!
  \sum\limits_{k=\min(m,j)}^{\max(m,j)-1}\abs{A_k}^{-1})G_{mj}(\zeta)}_{B(\mathcal{K})}
\le C_a\,,
  \end{equation*}
where $\widetilde{\gamma}(\zeta)$ is given by \eqref{eq:gamma-tilde}.
\item \label{in-spectrum-cor-commutation-finite}
If $\zeta\in\sigma_p(J)\cap (r,s)$ and $\ker(A_m)$ is trivial for any
$m\in\nats$, then
  \begin{equation*}
    \norm{\exp(\left(\frac12-\epsilon\right)w(\zeta)\sum_{k=1}^{m-1}
\abs{A_k}^{-1})u_m}_{\mathcal{K}}\le
   C_b\,,
  \end{equation*}
where
 $u$ is the corresponding
  eigenvector.
\end{enumerate}
The constant $C_a$ does not
depend on $m$ and $j$, and $C_b$ does not depend on $m$.
\end{corollary}
\begin{proof}
  We prove the claim in
  \eqref{not-in-spectrum-cor-commutation-finite}. The statement in
  \eqref{in-spectrum-cor-commutation-finite} is proven
  analogously. First one uses the argumentation of
  Corollary~\ref{cor:norm-to-infty-finite} to simplify the expression
  of $\gamma(\zeta)$ when $\delta$ is large enough. Thus
\begin{equation*}
    \norm{\exp\left(\widetilde\gamma\msum_{m=1}^{j-1}
\phi_\delta(\abs{A_m})\right)
G_{jk}(\zeta)
\exp\left(-\widetilde\gamma\msum_{m=1}^{k-1}\phi_\delta(\abs{A_m})\right)}_{B(\mathcal{K})}
\le C
\end{equation*}
for all $j,k<N$. By letting $N\to\infty$, one obtains from this that
    \begin{equation*}
 \norm{\exp(\widetilde\gamma(\zeta)\mkern-12mu\nsum\limits_{k=\min(m,j)}^{\max(m,j)-1}
\mkern-12mu\phi_\delta(\abs{A_k}))G_{mj}(\zeta)}_{B(\mathcal{K})}
\le C\,.
\end{equation*}
Finally, using the argumentation in \eqref{eq:argumentation-cor}, one
concludes that, for any $p\gg 1$,
\begin{equation*}
  \msum_{k=1}^p\phi_\delta(\abs{A_k})-\msum_{k=1}^p \frac{1}{\abs{A_k}}
\end{equation*}
is a bounded operator whose norm is independent of $p\gg 1$.
\end{proof}
\section{Examples}
\label{sec:examples}
In this section, we consider concrete realizations of self-adjoint block Jacobi
operators having finite gaps in the essential spectrum. The examples
admit a straightforward calculation of the decay of the corresponding
generalized eigenvectors although in some cases it is somehow
involved. The estimates obtained in this way are compared with the ones
given by Theorem~\ref{thm:finite-interval}.
\\[5mm]
\noindent\textbf{Example 1.} Let us first consider the block Jacobi
operator given in
Definition~\ref{def:j-nought} so that, for any $n\in\nats$, $B_n=0$  and
$A_{n}=A_{n}^{(0)}$, where
\begin{equation}
 \label{eq:A-n-example1}
  A_n^{(0)}:=
  \begin{pmatrix}
    0&\lambda_n\\
    0 & 0
  \end{pmatrix}
\,,\qquad n\in\nats\,.
\end{equation}
Here the sequence $\{\lambda_n\}_{n=1}^\infty$ of complex numbers is
such that $\lambda_{n}\convergesto{n}\infty$. We denote this block Jacobi
operator by $J_0$. It can be decomposed as an infinite
orthogonal sum of matrices:
\begin{equation*}
  J_{0}=
  \begin{pmatrix}
    0&\lambda_1\\
    \cc{\lambda}_1&0
  \end{pmatrix}
\oplus
\boldsymbol{0}
\oplus
  \begin{pmatrix}
    0&\lambda_2\\
    \cc{\lambda}_2&0
  \end{pmatrix}
\oplus  \begin{pmatrix}
    0&\lambda_3\\
    \cc{\lambda}_3&0
  \end{pmatrix}
\oplus\dots\,,
\end{equation*}
where the matrix $\boldsymbol{0}$ is a one-dimensional matrix.
From this decomposition, one can deduce that $J_0$ is self-adjoint and
$\sigma(J_{0})$ is discrete and the Green function is a band matrix
function so
that its entries decrease faster than any nonfinite sequence. If,
instead of \eqref{eq:A-n-example1}, one assumes
$A_{n}=A_{n}^{(\epsilon)}$, where
\begin{equation*}
  A_n^{(\epsilon)}:=
  \begin{pmatrix}
    \epsilon_n&\lambda_n\\
    0 & \epsilon_n
  \end{pmatrix}
\,,\qquad n\in\nats\,,
\end{equation*}
with $\epsilon_n>0$ for all $n\in\nats$ and
$\epsilon_{n}\convergesto{n}0$, then the
corresponding block Jacobi operator given by
Definition~\ref{def:j-nought} and denoted by $J$,
is also self-adjoint and has discrete spectrum. (Note that $J$ and
$J_0$ could be examples of self-adjoint block Jacobi operators not
satisfying the Carleman criterion when the sequences
$\{\lambda_n\}_{n=1}^\infty$ and $\{\epsilon_n\}_{n=1}^\infty$ are
chosen appropriately). The Green function corresponding to $J$ is not
anymore a band matrix. Nevertheless, the  entries of the Green
matrices also decay as fast as the sequence
$\{\epsilon_n\}_{n=1}^\infty$ permits. This is shown by the following argument.
By the
Hilbert second resolvent identity, one has, for $\zeta\in\rho(J_0)\cap\rho(J)$,
\begin{equation*}
  \inner{\left[(J_0-\zeta I)^{-1}-(J-\zeta
      I)^{-1}\right]P_ju}{P_ku}=
  \inner{J_{\epsilon}(J_0-\zeta
      I)^{-1}P_ju}{(J-\zeta I)^{-1}P_ku}\,,
\end{equation*}
where $J_{\epsilon}$ is the block Jacobi operator given by
Definition~\ref{def:j-nought} with $B_n=0$ and
$A_{n}=\boldsymbol{\epsilon}_{n}$, where
\begin{equation*}
  \boldsymbol{\epsilon}_n=\begin{pmatrix}
    \epsilon_n&0\\
    0 & \epsilon_n
  \end{pmatrix}
\end{equation*}
for any $n\in\nats$. Note that  $J_{\epsilon}$ is compact and,
moreover, the rank of $J_{\epsilon}(J_{0}-\zeta I)^{-1}P_{j}$ is
finite since all the nonzero entries of its matrix representation are
in a finite vicinity of the blocks indexed by the value of $j$
(denoted $V(j)$). Therefore
\begin{equation*}
  \norm{J_{\epsilon}(J_{0}-\zeta I)^{-1}P_{j}}\le\max_{i\in
    V(j)}\{\epsilon_{i}\}\norm{(J_{0}-\zeta I)^{-1}}\,.
\end{equation*}

Taking into account that $\norm{A_n}=\abs{\lambda_n}(1+o(1))$ as
$n\to\infty$, one obtains from  Theorem~\ref{thm:finite-interval} that
    \begin{equation}
      \label{eq:estimate-example1}
      \norm{G_{mj}(\zeta)}\le
    C
   \exp(-\widetilde{\gamma}(\zeta)\!\!\sum\limits_{k=\min(m,j)}^{\max(m,j)-1}
1/\abs{\lambda_k})\,,
\end{equation}
where $\widetilde{\gamma}(\zeta)$ is given in \eqref{eq:gamma-tilde}.
If $m$ or $j$ tend to $\infty$ and the corresponding series is
divergent, then the Green matrix elements $G_{mj}(\zeta)$ could
decrease faster than any power when one let either one of the indices
$m,j$ grow and the other is kept fixed.

Note that the estimate in \eqref{eq:estimate-example1} does not
contain any information on the sequence
$\{\epsilon_{n}\}_{n=1}^{\infty}$. At the same time, the previous
considerations show that the optimal estimates depend on $\max_{i\in
    V(j)}\{\epsilon_{i}\}$. Hence, for some choices of the sequence
  $\{\lambda_{n}\}_{n=1}^{\infty}$ the estimate in
  \eqref{eq:estimate-example1} may be close to optimal, however, in
  most cases, it is far from the real estimates. Example 1 illustrates
  a case when our results could not be optimal.
\\[5mm]
\noindent\textbf{Example 2.} Let
\begin{equation}
 \label{eq:A-n-example2}
  A:=
  \begin{pmatrix}
    1& x\\
    0 & 1
  \end{pmatrix}\,,\quad x\in\reals\,.
\end{equation}
Define the sequences of $2\times2$ matrices $\{B_n\}_{n\in\nats}$ and
$\{A_n\}_{n\in\nats}$ so that $B_n=0$ and $A_n=A$. Let $J$ be the
operator in $l_2(\nats,\complex^2)$ given by
Definition~\ref{def:j-nought}.  Since the elements of the sequences
$\{B_n\}_{n\in\nats}$ and $\{A_n\}_{n\in\nats}$ are all equal to
constant matrices, the matrix \eqref{eq:block-jm} is periodic. Let us
find conditions on $x$ for the operator $J$ to have a finite gap in
the essential spectrum.
\begin{proposition}
  \label{prop:gap-essential-2}
  For the operator $J$ given in this example,
  $\sigma_{ess}(J)=[-2+x,2+x]\cup[-2-x,2-x]$. Thus, if $\abs{x}>2$, then
  the interval $(2-\abs{x},-2+\abs{x})$ is a gap in the essential
  spectrum of $J$.
\end{proposition}
\begin{proof}
  Denote by $\mathbb{T}$ the unit circle in the complex plane and by
  $\mu$ the Lebesgue measure on the unit circle normalized so that
  $\mu(\mathbb{T})=1$. We also consider the Hilbert space $l_2(\ints,
  \complex^d)$ and denote any of its elements, \ie the sequence
  $w=\{w_k\}_{k\in\ints}$ (($w_k\in\complex^d$ for all $k\in\ints$) by
  an infinite column vector $(\dots,w_{-1},w_0,w_1,\dots)^{\intercal}$ (\cf
  Section~\ref{sec:notation}).

 Define in $l_2(\ints,\complex^d)$, the map $V$ by
\begin{equation*}
    (Vw)(z)=\sum_{n=-\infty}^\infty \exp(-\I k\theta)w_k\,,
      \end{equation*}
where $z=e^{\I \theta}$, $\theta\in[0,2\pi)$. It is known that this is
a unitary map from $l_2(\ints, \complex^d)$ onto
$L_2(\mathbb{T},\complex^d,\mu)$. Moreover (\cf
  \cite[Chap.\,2]{MR1071374})
any ``double-infinite'' Jacobi operator $J_{\ints}$ with constant
block entries acting in
$l_2(\ints, \complex^d)$ is transformed under this map into the multiplication
operator by certain matrix function in $L_2(\mathbb{T},\complex^d,
\mu)$, namely,
\begin{equation*}
  VJ_\ints V^*=M_{\phi}\,.
\end{equation*}
Here the multiplication operator $M_{\phi}$ is defined by
\begin{equation*}
  (M_\phi f)(z):=\phi(z)f(z)\quad\text{ for all } f\in L_2(\mathbb{T},\complex^d,
\mu)\,,
\end{equation*}
where $\phi$ is a $d\times d$-matrix function which will be determined
some lines below.

Let $J_\ints$ be the operator corresponding to the matrix
\begin{equation*}
  \begin{pmatrix}
\ddots&\ddots&\ddots& & & \\
\ddots&0&A&0& & \\
\ddots&A^*&0&A&0& \\
&0&A^*&0&A&\ddots\\
& &0&A^*&0&\ddots\\
& & &\ddots&\ddots&\ddots
\end{pmatrix}\,.
\end{equation*}
If $w$ is such that $w_k=0$ for all $k\ne j$, then
\begin{align*}
  VJ_{\ints}w&=J_\ints(\dots,0, 0,w_{j},0,0,\dots)^\intercal
=V(\dots,0, Aw_{j},\kern-1.5em
\underset{\underset{j-\text{th}\kern.2em\text{position}}{\uparrow}}{0}
\kern-1.5em,A^*w_{j},0,\dots)^\intercal\\[2mm]
&=e^{-\I(n-1)\theta}Aw_j+e^{-\I(n+1)\theta}A^*w_j=(e^{\I\theta}A+e^{-\I\theta}A^*)e^{-\I n\theta}w_j \\
&=(e^{\I\theta}A+e^{-\I\theta}A^*)Vw\,.
\end{align*}
Thus, $J_\ints$ is transformed by $V$ into the operator of
multiplication $M_\phi$ with $\phi(z)=zA+\cc{z}A^*$.

For each $z\in\mathbb{T}$, the spectrum of $\phi(z)$ is
\begin{equation*}
  \{z+\cc{z}+x,z+\cc{z}-x\}\,.
\end{equation*}
According to \cite[Chap.\,8 Sec.\,4]{MR1192782}, the spectrum of
$M_\phi$ is given by the union of the essential range of the
eigenvalues as functions of $z$. Therefore
\begin{equation*}
  \sigma(J_{\ints})=[-2+x,2+x]\cup[-2-x,2-x]=\sigma_{ess}(J_\ints)\,.
\end{equation*}
Now, if $K$ is given by the matrix
\begin{equation*}
  \begin{pmatrix}
\ddots&\ddots&\ddots& & & \\
\ddots&0&0&0& & \\
\ddots&0&0&A&0& \\
&0&A^*&0&0&\ddots\\
& &0&0&0&\ddots\\
& & &\ddots&\ddots&\ddots
\end{pmatrix}\,,
\end{equation*}
then $J_\ints=J\oplus J_{\rm up}+K$, where $J_{\rm up}$ is the
operator associated with the matrix
\begin{equation*}
  \begin{pmatrix}
\ddots&\ddots& &  \\
\ddots&0&A&  \\
&A^*&0&A \\
&&A^*& 0\\
\end{pmatrix}
\end{equation*}
in the subspace  $l_2(\ints_-, \complex^2)$ ($\ints_-=\{\dots,-1,0\}$)
of the Hilbert space of $l_2(\ints,\complex^2)$. Since $K$ is a finite
rank operator, it follows from Weyl theorem (see \cite[Thm.\,3 Sec.\,1
Chap.\,9]{MR1192782}) that
\begin{equation*}
  \sigma_{ess}(J_\ints)=\sigma_{ess}(J)\cup\sigma_{ess}(J_{\rm up})\,.
\end{equation*}
Hence, if one shows that
$\sigma_{ess}(J_\ints)\subset\sigma_{ess}(J)$, then
$\sigma_{ess}(J_\ints)=\sigma_{ess}(J)$ and the assertion of
the proposition follows.

The fact that $\lambda\in\sigma_{ess}(J_\ints)$ is equivalent to the
existence of a Weyl
sequence (also known as singular sequence) at $\lambda$ \cite[Thm.\,2 Sec.\,1
Chap.\,9]{MR1192782}, \ie there is $\{v(m)\}_{m=1}^\infty$ such that
\begin{enumerate}[(1)]
\item $\inf_{m\in\nats}\norm{v(m)}>0$\label{first-weyl}
\item $v(m)\wconvergesto{m} 0$ \label{second-weyl}
\item $(J_\ints -\lambda I)v(m)\convergesto{m} 0$\label{third-weyl}\,.
\end{enumerate}
It follows from \eqref{first-weyl} that there is a $p\in\ints$ such
that for all $m\in\nats$,
\begin{equation*}
  \sum_{k=p}^\infty\norm{v_k(m)}^2>0\,.
\end{equation*}
The sequence
  $\{v_{j+k}(m)\}_{j\in\ints}$ in $l_2(\ints,\complex^2)$
satisfies \eqref{first-weyl}--\eqref{third-weyl} due to the fact that
the $J_\ints$ have constant block coefficients. Therefore, if
$P_\nats$ is the projection in $l_2(\ints,\complex^2)$ onto the
subspace $l_2(\nats,\complex^2)$, then
\begin{equation*}
  P_\nats\{v_{j+k}(m)\}_{j\in\ints}
\end{equation*}
is a Weyl sequence at $\lambda$ for the operator $J$, which in turn
means that $\lambda\in\sigma_{ess}(J)$.
\end{proof}
In view of Proposition~\ref{prop:gap-essential-2}, operator $J$
satisfies Hypothesis~\ref{hyp:main}. Let us find the estimates given
by Theorems~\ref{thm:finite-interval} and
\ref{thm:finite-interval-in-spectrum} for this operator. To this end,
one first computes the norm of the matrix $A$ by calculating the
largest eigenvalue of $\abs{A}$. One has
\begin{equation}
  \label{eq:norm-A}
  \norm{A}^{2}=1+\frac{\abs{x}^2}{2}
+\sqrt{\left(1+\frac{\abs{x}^2}{2}\right)^2-1}\,.
\end{equation}
Inserting this expression into the formula given in
Theorem~\ref{thm:finite-interval}, one arrives at the following
estimate. If $\zeta$ is an eigenvalue and  $u$ is the corresponding
eigenvector, then
\begin{equation}
 \label{eq:eigenvalue-example2}
  \norm{u_m}_{\mathcal{K}}\le C\exp(-\gamma(\zeta)m\left(1+\frac{\abs{x}^2}{2}
+\sqrt{\left(1+\frac{\abs{x}^2}{2}\right)^2-1}\right)^{-1/2})\,,
\end{equation}
where the function $\gamma$ is given by
\eqref{eq:gamma-noncommut-small-imaginary}. If $\zeta$ is not an
eigenvalue, then
  \begin{equation*}
      \norm{G_{mj}(\zeta)}_{B(\mathcal{K})}\le
    C\exp(-\frac{\gamma(\zeta)\abs{j-k}}{\left(1+\frac{\abs{x}^2}{2}
+\sqrt{\left(1+\frac{\abs{x}^2}{2}\right)^2-1}\right)^{1/2}})\,.
  \end{equation*}
This result can be compared with the straightforward computation of
the generalized eigenvectors by the so-called transfer matrices.
By defining
\begin{equation}
 \label{eq:transfer-matrices}
  M_n(\zeta):=
  \begin{pmatrix}
    0& I\\
-A_n^{-1}A_{n-1}^* & \zeta A_n^{-1}
  \end{pmatrix}\,,
\end{equation}
the recurrence equation $\Upsilon u=zu$ (see
\eqref{eq:difference-expr}) can be written as
$\widetilde{u}_{n+1}=M_n\widetilde{u}_{n}$ for $n>1$, where
\begin{equation}
  \label{eq:u-tilde-n}
  \widetilde{u}_n=
  \begin{pmatrix}
    u_{n-1}\\ u_n
  \end{pmatrix}\,.
\end{equation}
Thus, estimates of the products of the transfer matrices yield decay
estimates for generalized eigenvectors.
Since in this case $A_{n}=A$ for any
$n\in\nats$, the matrix $M_{n}(\zeta)$ does not depend on $n$ and will
be denoted by $M(\zeta)$.
The eigenvalues of $M(\zeta)$ for a fixed
$\zeta$ are the
solutions with respect to $\mu$ of the equation
\begin{equation*}
 \det\left(A^{*}-\zeta\mu I +\mu^{2}A\right)=0\,.
\end{equation*}
Therefore the four eigenvalues of $M(\zeta)$ are
\begin{equation*}
  \mu=\frac{\zeta\mp x}{2}\pm\sqrt{\left(\frac{\zeta\mp x}{2}\right)^{2}-1}\,.
\end{equation*}
By a straightforward computation, the minimal decay of an
eigenvector in the gap $(2-\abs{x},-2+\abs{x})$, $\abs{x}>2$ is
\begin{equation*}
  \left(\abs{\frac{\abs{\zeta}-\abs{x}}{2}}+\sqrt{\left(\frac{\abs{\zeta}-\abs{x}}{2}\right)^{2}-1}\right)^{-m}\,.
\end{equation*}
Now assume that $\zeta$ is placed near the edge of the gap given in
Proposition~\ref{prop:gap-essential-2}, say
$\zeta=2-\abs{x}+\epsilon$, ($0<\epsilon<<1$). In this case, the
minimal decay of an eigenvector corresponding to this $\zeta$ is
\begin{equation*}
  \left(1+\epsilon/2+\sqrt{(1+\epsilon/2)^{2}-1}\right)^{-1}
\end{equation*}
Thus, as $\epsilon\to 0$, the minimal decay of the corresponding eigenvector is
\begin{equation*}
  \left(1+\sqrt{\epsilon}+O(\epsilon)\right)^{-m}=\exp(-m[\sqrt{\epsilon}+O(\epsilon)])\,.
\end{equation*}
Compare this with \eqref{eq:eigenvalue-example2}, where in this case
$\gamma(\zeta)\simeq \sqrt{\epsilon}(2\abs{x}-4)$. Note that in both
cases the coefficient determining the decay rate is determined by
$\sqrt{\epsilon}$, \ie by the square root of the distance to the edge
of the gap in the essential spectrum when this distance is small.
\\[.5cm]
\noindent\textbf{\large Example 3.}
Let $J$ be the operator given in
Definition~\ref{def:j-nought} with $B_n=0$ for any
$n\in\nats$. Consider the constant matrix $A$ given in \eqref{eq:A-n-example2}
and define
\begin{equation}
  \label{eq:a-n-example-3}
 A_n:=(n^{\alpha}+c_n)A\,,
\end{equation}
where $\alpha\in(1/2,1)$,
$c_{2n-1}=c_1$, and $c_{2n}=c_2$ ($c_1,c_2\in\reals$). For this
example, it is assumed that $\abs{x}<2$. This assumption, as seen
below, guarantees the existence of a bounded gap in the essential
spectrum of $J$.

The block Jacobi operator $J$ exhibits a gap in the essential
spectrum. This operator does not reduce to ``scalar'' Jacobi operators
and its spectral analysis requires, as shown in
\cite{inpreparation-naboko-silva-18}, generalizing some of the
techniques used for studying Jacobi operators.

Due to the 2-periodic character of the matrix weights, we first
construct block of transfer matrices (see for instance
\cite{MR1959871}, and \cite{MR3236260} for the block version). Using
the matrices given in \eqref{eq:transfer-matrices}, define the
monodromy matrix
\begin{equation}
  \label{eq:monodromy-example-3}
W_n(\zeta):=M_{2n}(\zeta)M_{2n-1}(\zeta)
\end{equation}
for each $n\in\nats$. Then
\begin{equation}
  \label{eq:linear-system}
  \widetilde{u}_{2n+1}=W_n(\zeta)\widetilde{u}_{2n-1}\,,
\end{equation}
where $\widetilde{u}_{n}$ is given by \eqref{eq:u-tilde-n}.  Thus, the
generalized eigenvectors of $J$ at the spectral parameter $\zeta$ are
solutions to the discrete linear system \eqref{eq:linear-system} and,
by the same token, the spectral properties of $J$ are determined by
this system.  We use here an approach to the analysis of
\eqref{eq:linear-system} which has a heuristic component and refer the
reader to \cite{inpreparation-naboko-silva-18} for the complete proof.

Substituting \eqref{eq:transfer-matrices} into
\eqref{eq:monodromy-example-3}, one obtains
\begin{equation*}
  W_{n}(\zeta)=
  \begin{pmatrix}
    -A_{2n-1}^{-1}A_{2n-1}^{*} & \zeta A_{2n-1}^{-1}\\
    -\zeta A_{2n}^{-1}A_{2n-1}^{-1}A_{2n-2}^{*}
    & -A_{2n}^{-1}A_{2n-1}^{*}+\zeta^{2}A_{2n}^{-1}A_{2n-1}^{-1}
  \end{pmatrix}\,.
\end{equation*}
Taking into account \eqref{eq:a-n-example-3}, one verifies by
straightforward calculations that the monodromy matrix can be written
as follows
\begin{equation}
  \label{eq:decomposition-monodromy}
  W_{n}(\zeta)=\left(1-\frac{\alpha}{2n}\right)
  \left[
\begin{pmatrix}
  -\left(1+\frac{c_{2}-c_{1}}{(2n)^{\alpha}}\right)A^{-1}A^{*}
& \frac{\zeta}{(2n)^{\alpha}} A^{-1}\\[1mm]
-\frac{\zeta}{(2n)^{\alpha}} A^{-2}A^{*}
& -\left(1-\frac{c_{2}-c_{1}}{(2n)^{\alpha}}\right)A^{-1}A^{*}
  \end{pmatrix}
  +\Gamma_{n}
  \right]\,,
\end{equation}
where $\Gamma_{n}$ is such that the sequence
$\{\norm{\Gamma_{n}}\}_{n=1}^{\infty}$ is summable. In the
asymptotic analysis of \eqref{eq:linear-system}, the sequence
$\{\Gamma_{n}\}_{n\in\nats}$ is not relevant and the factor
$1-\alpha/(2n)$ can be easily deal with at the end of the
computation.

Fix arbitrary complex numbers $\omega$ and $\zeta$, and put
$\epsilon:=(2n)^{-\alpha}$ with fixed $n\in\nats$. Let us compute the
determinant of the matrix $\mathcal{M}-\omega I$, where
\begin{equation*}
  \mathcal{M}=\mathcal{M}(\epsilon,\zeta):=
  \begin{pmatrix}
  -\left(1+\epsilon(c_{2}-c_{1})\right)A^{-1}A^{*}
& \epsilon\zeta A^{-1}\\[1mm]
-\epsilon\zeta A^{-2}A^{*}
& -\left(1-\epsilon(c_{2}-c_{1})\right)A^{-1}A^{*}
  \end{pmatrix}
\end{equation*}
(compare this expression with \eqref{eq:decomposition-monodromy}).
To this end, consider the $2\times 2$ auxiliary matrix
\begin{equation*}
  \gamma_{\pm}(\epsilon,\omega):=-\left(1\pm
    \epsilon(c_{2}-c_{1})\right)A^{-1}A^{*}-\omega I\,.
\end{equation*}
By using the Schur complement for computing
$\det(\mathcal{M}-\omega I)$ (see \cite{zhang-2005}), one obtains
\begin{align*}
  \det(\mathcal{M}-\omega I)
  &=\det\gamma_{-}\det\left[\gamma_{+}
    +(\epsilon\zeta)^{2} A^{-1}\gamma_{-}^{-1}A^{-2}A^{*}\right]\\
  &=\det\gamma_{-}\det\left[\gamma_{+}
    +\frac{1}{\det\gamma_{-}}(\epsilon\zeta)^{2}
    A^{-1}\widetilde\gamma_{-}A^{-2}A^{*}\right]\\
  &=\frac{1}{\det\gamma_{-}}\det\left[\left(\det\gamma_{-}\right)\gamma_{+}
    +(\epsilon\zeta)^{2} A^{-1}\widetilde\gamma_{-}A^{-2}A^{*}\right]\,,
\end{align*}
where $\widetilde\gamma_{-}$ is the adjugate of the matrix
$\gamma_{-}$. Substituting \eqref{eq:A-n-example2} into the last
expression and performing all necessary elementary (though lengthy)
computations, one arrives at
\begin{equation}
  \label{eq:computed-determinant}
  \begin{split}
  \det(\mathcal{M}-\omega I)&=\det\gamma_{-}\det\gamma_{+}+
  (\epsilon\zeta)^{2}\bigl[(1-\epsilon^{2}(c_{2}-c_{1})^{2})
  (x^{2}+2)\\ &+\omega(4-x^{4}\epsilon(c_{2}-c_{1}))-\omega^{2}(3x^{2}-2)
  \bigr] +(\epsilon\zeta)^{4}C\,,
\end{split}
\end{equation}
where the scalar $C$ depends only on $x,c_{1},c_{2}, \epsilon$, and its value,
being a polynomial in $\epsilon$, has no effect in the remaining
computations. Recall that $\det\gamma_{\pm}$ are polynomials of
$\omega$ and $\epsilon$ in both variables of degree 2.

Let $\mu_{\pm}$ be the eigenvalues of
\begin{equation*}
  -A^{-1}A^{*}=
  \begin{pmatrix}
    x^{2}-1 & x\\
    -x & -1
  \end{pmatrix}\,.
\quad \text{ Thus, }\quad
  \mu_{\pm}=\left(\frac{x^{2}}{2}-1\right)\pm
\sqrt{\left(\frac{x^{2}}{2}-1\right)^{2}-1}\,.
\end{equation*}
Note that each eigenvalue $\mu_{\pm}$ is a multiplicity two eigenvalue
of the $4\times 4$ matrix $\mathcal{M}$ when $\epsilon=0$. Now, on the
basis of \cite[Chap.\,2 Sec.\,2]{MR0407617}, since the algebraic and geometric
multiplicities of the eigenvalues of $\mathcal{M}$ coincide, the
following asymptotic ansatz
\begin{equation*}
  \omega=\mu_{\pm}(1+\epsilon\rho+O(\epsilon^{2}))\,,\quad\text{ as }
  \epsilon\to\infty\,,
\end{equation*}
can be substituted into \eqref{eq:computed-determinant} to find $\rho$
from the equation $\det(\mathcal{M}-\omega I)=0$ by equating
coefficients of the powers of $\epsilon$. From the equation
corresponding to $\epsilon^{2}$, one obtains
\begin{equation*}
  \rho=\pm\sqrt{(c_2-c_1)^2-\zeta^2\left(1-
\frac{\abs{x}^2}{4}\right)^{-1}}=:\rho_{\pm}\,.
\end{equation*}
Hence, depending on the value of $\zeta$, both $\rho_{\pm}$ are either
real or pure imaginary.  Note that, in the leading approximation, the
system is always in the elliptic regime ($\abs{\mu_{\pm}}=1$) whereas,
in the second approximation, which corresponds to the values of
$\rho_{\pm}$, its character depends on the spectral parameter
$\zeta$. Indeed, if $\zeta$ is in
\begin{equation}
\label{eq:gap-example3}
\left(-\abs{c_2-c_1}\sqrt{1-
x^2 /4}\,,\quad \abs{c_2-c_1}\sqrt{1-
x^2 /4}\right)\,,
\end{equation}
then the system is secondary hyperbolic ($\rho_{\pm}$ are real). If
$\zeta$ is not in the closure of the interval given in
\eqref{eq:gap-example3}, then the system is secondary elliptic
($\rho_{\pm}$ are purely imaginary). The fact that the character of
the leading approximation is elliptic means that the asymptotic
behaviour of the solutions, namely their growth or decay, is actually
determined by the sub-leading coefficients of the approximation, \ie,
the values of $\rho$. This is related to the divergence of the series
defined by the sequence $\{1/(2n)^{\alpha}\}_{n\in\nats}$ (recall that
$1/2<\alpha<1$).  Note that the two degenerated eigenvalues of
$\mathcal{M}$ when $\epsilon=0$ split into two pairs of complex
conjugate simple eigenvalues. In the secondary elliptic case, the
splitting goes tangentially to the unit circle, while in the secondary
hyperbolic case the splitting goes along the radius (perpendicular to
the circle).

Now, one performs the asymptotic analysis of \eqref{eq:linear-system}
\emph{\`{a} la} Levinson (see for instance \cite{MR2943838}). This allows
one to conclude that the interval \eqref{eq:gap-example3} is a gap in
the essential spectrum of $J$. Indeed, the idea of the Levinson
approach is to replace the nontrivial structure of the solutions to
\eqref{eq:linear-system} by the product of the eigenvalues of
$\mathcal{M}$ and then multiply by
\begin{equation*}
  \prod_{k=1}^{n}(1-\frac{\alpha}{2k})\asymp n^{-\alpha/2}\,.
\end{equation*}
It is worth remarking that the formal application of the Levinson
approach is the heuristic part of our analysis. The reason of this is
that the product of eigenvalues of the monodromy matrix usually gives
the correct main exponential term. However, the power in $n$ factor is
not always correct.

The product of the eigenvalues of $\mathcal{M}$ yields (up to a
constant factor) the following four expressions
\begin{equation}
 \label{eq:see-behaviour}
  \mu_{+}^{n}\exp\left(\rho_{\pm}\sum_{k=1}^{n}\frac{1}{(2k)^{\alpha}}\right)\,,
\qquad \mu_{-}^{n}\exp\left(\rho_{\pm}\sum_{k=1}^{n}\frac{1}{(2k)^{\alpha}}\right)\,.
\end{equation}

Thus the heuristic application of the Levinson approach allows one to
conclude that there are
two linearly independent solutions $u^{\pm}$ such that
\begin{equation*}
  \norm{u_{2n}^{\pm}}\asymp\frac{1}{n^{\alpha/2}}\abs{\exp
\left\{\sum_{k=n_0}^n\frac{\rho_{\pm}}{(2n)^\alpha}\right\}}\,.
\end{equation*}
Since $\rho_{\pm}$ are purely imaginary when the spectral parameter
$\zeta$ is outside the interval in \eqref{eq:gap-example3}, the
exponential factor in the last formula has modulus 1. In this case,
all the generalized eigenvectors decay as $n^{-\alpha/2}$ and the role
of the perturbation $\epsilon$ is bounded by a purely oscillating
factor. This behaviuor corresponds to the essential spectrum in the
region given by \eqref{eq:gap-example3}. A rigorous proof of this fact
can be obtained using Weyl sequences constructed on the basis of the
asymptotic properties of \eqref{eq:see-behaviour} (see
\cite{MR2480099}).  On the other hand, in the secondary hyperbolic
case (i.e inside the interval given by \eqref{eq:gap-example3}),
$\rho_{\pm}$ are real with opposite signs which implies that one
solution of \eqref{eq:linear-system} grows while the other
decays. This case corresponds, at least at the physical level, to the
absence of the essential spectrum (for a rigorous proof see the
methods used in \cite{MR2558158}).

A detailed asymptotic analysis of \eqref{eq:linear-system} and the
proof of the above assertions concerning the spectral properties of
$J$ in this example are done in \cite{inpreparation-naboko-silva-18}
in which a further development of the techniques used in
\cite{MR2579689} is carried out.

The estimates for the decay of generalized eigenvalues inside the gap
are thus given by
\begin{align*}
  \norm{u_{2n}}
\asymp \frac{1}{n^{\alpha/2}}\exp
\left\{\frac{\rho_{-}n^{1-\alpha}}{2^\alpha(1-\alpha)}\right\}\,,\qquad
\rho_{-}<0\,.
\end{align*}
In particular, this implies that
\begin{equation}
\label{eq:something-to-compare}
  \norm{u_n}\le \frac{C_0}{n^{\alpha/2}}\exp
\left\{\frac{\rho_{-}n^{1-\alpha}}{2(1-\alpha)}\right\}
\end{equation}
for some positive constant $C_{0}$.

Let us calculate the bound of the decay of the generalized
eigenvectors given by Theorem~\ref{thm:finite-interval} for this
case. To this end, we first compute the norm of $A_n$. We almost
have already done this, since we have \eqref{eq:norm-A} and therefore
\begin{equation*}
  \norm{A_n}= \abs{n^{\alpha}+c_n}\norm{A}\,.
\end{equation*}
Also the expression of $\gamma(\zeta)$ can be simplified according to
Corollary~\ref{cor:norm-to-infty-finite}. Assuming that $\zeta$ is in
the gap \eqref{eq:gap-example3},
one thus has
\begin{equation*}
  \gamma(\zeta)= w(\zeta)\left(\frac12-\epsilon'\right)\,.
\end{equation*}
Hence, for a generalized eigenvector $u$, there are constants $C,\,\widetilde{C}>0$
such that
\begin{align}
  \norm{u_{n}}&\le
  C\exp(-\gamma(\zeta)\nsum_{k=1}^{n-1}\norm{A_k}_{B(\mathcal{K})}^{-1})
\nonumber\\
&\le\widetilde{C}\exp(\left[\frac{1-
\abs{x}^2 /4}{1+\frac{\abs{x}^2}{2}
+\sqrt{\left(1+\frac{\abs{x}^2}{2}\right)^2-1}}\right]^{1/2}\frac{\rho_{-}n^{1-\alpha}}{1-\alpha}\left(\frac12-\epsilon'\right))\,.\label{eq:last-inequality}
\end{align}
Compare this inequality with \eqref{eq:something-to-compare}. When
$x\to 0$ (recall that $\abs{x}<2$), the expression in the square
brackets in \eqref{eq:last-inequality} goes to $1$ and therefore the
arguments of the exponential function in \eqref{eq:last-inequality}
and \eqref{eq:something-to-compare} exactly coincide up to $\epsilon'$
which is arbitrarily small.

\section{Discrete version of the method}
\label{sec:discr-vers-theor}

In this section, a different realization of the method used in the
previous sections is presented. Here we use an alternative expression
for the auxiliary operators \eqref{eq:def-tilde-Phi-m}. The result of
this change is an estimate that could be more precise in various
cases, in particular when the
sequence of operators $\{A_k\}_{k=1}^\infty$ has
multiple occurrences of $\norm{A_k}\ll 1$.

Redefine the sequence of operators $\{\Phi_m\}_{m=1}^{\infty}$ (see the
proof of Theorem~\ref{thm:finite-interval}) as follows:
\begin{equation}
\label{eq:def-tilde-Phi-md}
  \Phi_m:=
\begin{cases}
\mprod\limits_{k=1}^{m-1}
\left(
  1-\frac{\gamma}{\norm{A_k}}
\right)^{-1}I\,,&
m\le N\,,\\[4mm]
\mprod\limits_{k=1}^{N}
\left(
  1-\frac{\gamma}{\norm{A_k}}
\right)^{-1}I\,,&
m> N\,,
\end{cases}
\end{equation}
where $\gamma>0$. As in the proof of Theorem~\ref{thm:finite-interval},
one establishes the inequalities \eqref{eq:estimate-real-F} and
\eqref{eq:estimate-imaginary-F}. From these inequalities, one obtains
\begin{align}
  \label{eq:real-imaginary-discrete}
  \norm{\Re F}&\le \sup_{m\in\nats}\left\{\norm{A_{m}}
\left((1-\gamma/\norm{A_{m}})^{-1}+
(1-\gamma/\norm{A_{m}})-2\right)\right\}\\
\norm{\Im F}&\le \sup_{m\in\nats}\left\{\norm{A_{m}}
\left((1-\gamma/\norm{A_{m}})^{-1}-
(1-\gamma/\norm{A_{m}})\right)\right\}\,.
\end{align}
If $\norm{A_{m}}\convergesto{m} +\infty$, then
$\gamma/\norm{A_{m}}$ is arbitrarily small for $m$ sufficiently large. Let us assume,
without loss of generality, that $\gamma/\norm{A_{m}}<1$ for all
$m\in\nats$.

The sequences from the right-hand sides of the inequalities have the
form
\begin{align*}
  (1-x)^{-1}-2+(1-x)=(1-x)^{-1}x^{2}&=:\psi_{d}(x)\quad\text{ and}\\
 (1-x)^{-1}-(1-x)=(1-x)^{-1}(2-x)x&=:2\widetilde\psi_{d}(x)\,.
\end{align*}
The rational functions $\psi_{d}$ and $\widetilde\psi_{d}$ replace the
transcendental functions $\psi$ and $\widetilde\psi$ of the proof of
Theorem~\ref{thm:finite-interval}.
By a reasoning similar to the one used in the proof of
Theorem~\ref{thm:finite-interval} to obtain the estimates of the norm
of the real and imaginary part of the operator $F$, one concludes that
assigning
\begin{align*}
  \gamma &:=\delta\psi_{d}^{-1}\left(\frac{w^{2}(\Re
           \zeta)\epsilon}{2\delta(s-r)}\right)&\quad\text{ for
           }\norm{\Re F}\\
  \gamma &:=\delta\widetilde\psi_{d}^{-1}\left(\frac{w(\Re
           \zeta)(1-2\epsilon)}{2\delta}\right)&\quad\text{ for
           }\norm{\Im F}\,,
\end{align*}
one obtains the necessary estimates guaranteeing that the existing of
a constant $C=C(\zeta,J,\epsilon)>0$ such that
\begin{equation*}
  \norm{\Phi^{-1}_{N}(J-\zeta I)^{-1}\Phi_{N}}\le C
\end{equation*}
for $N$ sufficiently large. Thus we have given the sketch of the proof
of the following result.
\begin{theorem}
  \label{thm:last-one}
  Assume that Hypothesis~\ref{hyp:main} holds true and
  $\norm{A_{m}}\to\infty$ as $m\to\infty$. Fix an arbitrary $\delta>0$,
  $\epsilon\in(0,1/2)$, and $\eta\in(0,1)$.  If
  $\zeta\not\in\sigma(J)$ with $\Re\zeta\in(r,s)$ and $n_{0}$ is so
  large that $\gamma/\norm{A_{k}}< 1$ for $k\ge n_{0}$, then
    \begin{equation*}
      \norm{G_{mj}(\zeta)}_{B(\mathcal{K})}\le
    C \mprod_{k=\max\{\min(m,j),n_{0}\}}^{\max(m,j)-1}(1-\gamma/\norm{A_{k}})\,,
\end{equation*}
where
\begin{equation*}
 \gamma=\gamma(\zeta)=\min\left\{\delta\psi_{d}^{-1}\left(\frac{w^2(\Re\zeta)\epsilon}{2\delta(s-r)}\right),
\delta\widetilde{\psi}_{d}^{-1}\left(\frac{w(\Re\zeta)(1-2\epsilon)}{2\delta}\right)\right\}
\end{equation*}
when $\abs{\Im\zeta}\le
w(\Re\zeta)\frac{\epsilon}{2}$, and
\begin{equation*}
  \gamma=\gamma(\zeta)=\delta\widetilde{\psi}^{-1}_{d}\left(\frac{w(\Re\zeta)\epsilon(1-\eta)}{4\delta}\right)
\end{equation*}
otherwise. The constant $C$ depends neither on $m$
nor on $j$.
\end{theorem}

\begin{remark}
  \label{rem:discretemethod}
Note that the function $\psi$ given in \eqref{eq:psi-tilde-psi}
satisfies
\begin{equation*}
  \psi(x)=x^{2}(1+x+O(x^{2}))\quad\text{ as } x\to 0\,,
\end{equation*}
while
\begin{equation*}
  \psi_{d}=x^{2}(1+x+O(x^{2}))\quad\text{ as } x\to 0.
\end{equation*}
Therefore the first two terms of the expansion of $\psi$ and $\psi_{d}$
coincide as $x\to 0$. On the other hand, $\widetilde{\psi}$ (see
\eqref{eq:psi-tilde-psi}) obeys
\begin{equation*}
 \widetilde\psi(x)=x(1+x+O(x^{2}))\quad\text{ as } x\to 0
\end{equation*}
and
\begin{equation*}
  \widetilde\psi_{d}(x)=x(1+\frac{x}{2}+O(x^{2}))\quad\text{ as } x\to 0\,.
\end{equation*}
These asymptotic expansions show that there is an advantage in using
$\widetilde\psi_{d}$ instead of $\widetilde\psi$. It is possible to
obtain better estimates of the Green matrix entries by using the
discrete version of the method (see
\cite[Thm.\,5]{MR3028179}). Although the continuous version of the
method, which was presented in the previous sections, gives a more
convenient form for the estimates of the Green matrix entries, the
discrete version introduced in this section yields a slightly better
and more subtle form of the estimates. In general, the discrete
version of the method is more natural since the problem has itself a
discrete character.  Moreover, even in the case of noncommuting
entries of the block Jacobi matrix, the operator $\Phi_{m}$ can be
chosen as in \eqref{eq:def-tilde-Phi-md}, but with $\abs{A_{k}}$
instead of $\norm{A_{k}}$, $k\in\nats$. Of course in this case the
products should be taken in chronological order.  The exponential form
of $\Phi_{m}$ in the proof of Theorem~\ref{thm:finite-interval} has no
sense in that case. If we want to obtain a more precise estimate we
should choose this discrete version with nonscalar operators.  It
seems, however, that the continuous version of the method provides
estimates accurate enough for most of the applications.
\end{remark}

\section*{Acknowledgements}

The authors thank the anonymous reviewer for pertinent and useful
comments and remarks which led to an improved version of the manuscript.

S\,N was supported by grant RFBR 19-01-00657\,A (Sections 1--4) and
grant RScF 20-11-20032 (Sections 5--6). He expresses his
gratitute to the Institut Mittag-Leffler, where part of this work has
been done, for their kind hospitality and to the Knut and Alice Wallenberg
Foundation for the support given.  L\,O\,S has been supported by
UNAM-DGAPA-PAPIIT IN110818 and SEP-CONACYT CB-2015
254062. Part of this work was carried out while L\,O\,S was on
sabbatical leave at the University of Bath from UNAM with the support
of PASPA-DGAPA-UNAM.

\def\cprime{$'$} \def\lfhook#1{\setbox0=\hbox{#1}{\ooalign{\hidewidth
  \lower1.5ex\hbox{'}\hidewidth\crcr\unhbox0}}}

\end{document}